\documentclass[11pt,a4paper]{article}
\usepackage[utf8]{inputenc}
\usepackage[T1]{fontenc}

\usepackage{fourier} 

\usepackage{amsmath}
\usepackage{amsfonts}
\usepackage{amssymb}
\usepackage{amsthm}
\usepackage{ upgreek }

\usepackage{marginnote}

\usepackage{csquotes}

\usepackage{pdflscape}

\usepackage{stmaryrd}

\usepackage{enumitem}

\usepackage{nicefrac}

\usepackage{mathtools}

\usepackage{cleveref}
\usepackage{autonum} 
\usepackage{calc}

\usepackage{pgf,tikz}
\usetikzlibrary{matrix,automata,arrows,positioning,calc,patterns,fit,calc,
decorations.pathreplacing,decorations.pathmorphing,decorations.markings,matrix,
tikzmark,shadows.blur}
\usetikzlibrary{cd}
 \usepgflibrary{fpu}

\usepackage{filecontents}

\crefformat{equation}{(#2#1#3)}

\newtheorem{theorem}{Theorem}[section]
\newtheorem{proposition}[theorem]{Proposition}
\newtheorem{lemma}[theorem]{Lemma}
\newtheorem{corollary}[theorem]{Corollary}
\newtheorem{definition}[theorem]{Definition}

\theoremstyle{definition}
\newtheorem{examplex}[theorem]{Example}
\newenvironment{example}
  {\pushQED{\qed}\examplex}
  {\popQED\endexamplex}
\newtheorem{remarkx}[theorem]{Remark}
\newenvironment{remark}
  {\pushQED{\qed}\remarkx}
  {\popQED\endremarkx}

\numberwithin{equation}{section}

\newcommand{\Wick}[1]{{\mathpunct{:}#1\mathpunct{:}}}
\newcommand{\setsuch}[2]{\{#1\colon #2\}}

\input{macros.sty}
\input{graphs_gibbs_phi4.sty}
\input{diagrams_tikz.sty}

\newcommand{\unit}{\mathbf{1}}
\newcommand{\X}{\mathbf{X}}
\newcommand{\bk}{{\bs{k}}}
\newcommand{\bn}{{\bs{n}}}
\newcommand{\bm}{{\bs{m}}}
\newcommand{\PiNBPHZ}{\Pi_N^{\mathrm{BPHZ}}}
\DeclareMathOperator{\id}{id}

\begin{document}


\title{Perturbation theory for the $\Phi^4_3$ measure, \\
revisited with Hopf algebras}
\author{Nils Berglund, Tom Klose}
\date{5 October 2023}   

\maketitle

\begin{abstract}
\noindent
We give a relatively short, almost self-contained proof of the fact that the 
partition function of the suitably renormalised $\Phi^4_3$ measure admits an 
asymptotic expansion, the coefficients of which converge as the ultraviolet 
cut-off is removed. We also examine the question of Borel summability of the 
asymptotic series. The proofs are based on Wiener chaos expansions, 
Hopf-algebraic methods, and bounds on the value of Feynman diagrams 
obtained through BPHZ renormalisation. 
\end{abstract}

\leftline{\small 2010 {\it Mathematical Subject Classification.\/} 
60H15, 
35R11 (primary), 
81T17, 
82C28 (secondary). 
}



\section{Introduction} 
\label{sec:intro}

The $\Phi^4_d$ model, defined on the $d$-dimensional torus with 
$d\in\set{1,2,3}$, is probably one of the simplest non-trivial models in 
Euclidean quantum field theory. Here non-trivial means that the model can be 
proven to behave differently from a Gaussian field. In dimension $d=4$, it has 
been shown that the $\Phi^4_d$ model is indeed 
trivial~\cite{Aizenmann-Duminil-Copin}.

Despite it being simpler than other models, the analysis of the $\Phi^4$ model 
is by no means easy. The earliest works by Glimm and Jaffe and by Feldman 
approached the problem via a detailed combinatorial analysis of Feynman 
diagrams~\cite{Glimm_Jaffe_68,Glimm_Jaffe_73,Feldman74,Glimm_Jaffe_81}, 
entailing very long and technical proofs. Over the years, the analysis of the 
model has been gradually simplified. The works~\cite{BCGNOPS78,BCGNOPS_80} 
introduced the idea of using a renormalisation group approach, consisting in a 
decomposition of the covariance of the underlying Gaussian reference field into 
scales, which then allows to integrate sucessively over one scale after the 
other. This method was further perfected in~\cite{Brydges_Dimock_Hurd_CMP_95}, 
using polymers to control error terms, an approach based on ideas from 
statistical phy\-sics~\cite{Gruber_Kunz_71}. 

In another direction, the approach provided in~\cite{Brydges_Frohlich_Sokal_CPM83,Brydges_Frohlich_Sokal_CMP83_RW} allows to bound correlation functions (or $n$-point functions) without having to compute the partition function explicitly, by using it as a generating function. This involves the derivation of ske\-le\-ton inequalities, which were obtained up to third order in~\cite{Brydges_Frohlich_Sokal_CPM83}, and later extended to all orders in~\cite{Bovier_Felder_CMP84}. 
As an application, the latter work contains another proof for the asymptotic nature of perturbation theory for the~$\Phi^4_d$ model in dimensions~$d \in \{2,3\}$, a result that had earlier been obtained by different methods in~\cite{Feldman_Osterwalder_76} when~$d=3$ and~\cite{Dimock_1974} when~$d = 2$. In fact, based upon earlier work in~\cite{Glimm_Jaffe_Spencer_1974}, the article~\cite{Dimock_1974} establishes this result for more general polynomial~$P(\Phi)_2$ models.
A relatively compact derivation of bounds on the partition function based on the Bou\'e--Dupuis formula was recently obtained in~\cite{Barashkov_Gubinelli_18}. 
Besides, novel techniques based on singular stochastic PDEs have led to a new proof that the perturbative expansion for~$\Phi^4_2$ is asymptotic~\cite{Shen_Zhu_Zhu_23}.

In this review, we argue that there is still room for improvement in the analysis of the $\Phi^4_3$ model, taking advantage of quite recent developments in more algebraic approaches. 
%
We will present a rather compact argument which shows that, after sui\-ta\-ble renormalisation, all the terms in the perturbatively expanded partition function of the~$\Phi^4_3$ model are uniformly bounded in the cut-off parameter. This is the main contribution of our work.

It is well-known that~the perturbative series does \emph{not} converge~\cite{Jaffe_1965}, but that it can be resummed using the theory of Borel transforms; Sokal's theorem~\cite{Sokal80} provides two sufficient conditions under which that resummation procedure works. 
We will demonstrate that our setup provides a convenient framework to check one of these conditions, the \emph{remainder bounds}, in a systematic way.---However, we are only able to present a complete proof under a fairly strong moment bound assumption which seems highly non-trivial to check on its own part.
We \emph{do not} make any statement about (local) analyticity, the second condition, and instead only assume its validity.

Important sources of inspiration for our approach are the 
monograph~\cite{Peccati_Taqqu_book} by Peccati and Taqqu on Wiener chaos 
and cumulant expansions, the article~\cite{EFPTZ18} by Ebrahimi-Fard 
\textit{et al.} on deformations of Hopf algebras, and Hairer's  
overview~\cite{Hairer_BPHZ} of BPHZ renormalisation. 

This article is organized as follows. In Section~\ref{sec:set-up}, we introduce 
the set-up, including a definition of the renormalised $\Phi^4_3$ measure with 
cut-off $N$. In Section~\ref{sec:expansion}, we give a relatively concise proof 
of the fact that the partition function of the model admits a perturbative expansion in terms that converge as the cut-off $N$ is sent to infinity.
Finally, in Section~\ref{sec:Borel}, we address the question of Borel 
resummation of the per\-tur\-ba\-tive expansion and explain how our approach allows to systematize parts of its proof.

\smallskip

\textit{Acknowledgements:} Both authors thank the Erwin Schr\"odinger 
International Institute for Mathematics and Physics (ESI) of the University of 
Vienna for its kind hospitality and financial support during the Masterclass 
and Workshop of the Graduate School on “Higher Structures Emerging from 
Renormalisation” (8 to 19 November 2021).
NB was partly supported by the ANR project PERISTOCH, grant ANR--19--CE40--0023.
TK additionally thanks the Institut Denis Poisson in Orléans for its generous 
financial support and the warm hospitality during his visit in June 2022.
The article at hand has been completed while TK was employed at TU Berlin.
Both authors thank the two anonymous referees whose comments helped improve the presentation in this article.


\section{Set-up} 
\label{sec:set-up}

We are interested in the invariant measure of the massive $\Phi^4_3$ model on 
the torus \mbox{$\Lambda = \T^3 = (\R/\Z)^3$}, which can be formally defined as 
\begin{equation}
	\label{eq:phi4_measure} 
	\mu_{\Phi^4_3}(\6\phi) = \frac{1}{Z(\eps)}
	\exp\Biggset{-
		\int_\Lambda \Biggpar{\frac12 \norm{\nabla \phi(x)}^2 
			+ \frac{m^2}{2} \phi(x)^2 + \frac{\eps}{4}\phi(x)^4} \6x}
	\6\phi\;,
\end{equation}
where the partition function $Z(\eps)$ is the normalisation making 
$\smash{\mu_{\Phi^4_3}}$ a probability measure. In what follows, we will 
consider for convenience the case $m^2 = 1$. However, there is no difficulty in 
extending the results to any $m^2\geqs0$ by a Gaussian change of measure. In 
fact, even negative values of $m^2$ can be considered: they appear in the 
stochastic Allen--Cahn equation, see for instance 
\cite{BerglundDiGesuWeber,B-SPDE_book}.

As such, the measure~\eqref{eq:phi4_measure} is ill-defined, because there is 
no Lebesgue measure $\6\phi$ on $L^2(\Lambda)$. This issue can be solved in 
several 
steps, the first of which consists in considering a regularised version of the 
problem. Here it will be convenient to use a spectral Galerkin approximation 
with ultra-violet cut-off $N$. For $k\in\Z^3$, we write $e_k(x) = \exp(2\pi 
\icx k\cdot x)$ for the Fourier basis functions of $L^2(\Lambda)$, and set 
\begin{equation}
	\cH_N := \operatorname{span} \bigset{e_k \colon k \in \cK_N}\;, \qquad
	\cK_N := \bigset{k \in \Z^3 \colon \abs{k} := \abs{k_1} + \abs{k_2} + 
		\abs{k_3} 
		\leqs N}\;.
\end{equation} 
For any finite $N$, \eqref{eq:phi4_measure} defines a probability measure on 
$\cH_N$. In particular, the partition function can be written as 
\begin{equation}
	Z_N(\eps) 
	= Z_N(0)\;
	\Biggexpecin{\mu_N}{\exp\biggset{-\frac\eps4\int_\Lambda \phi(x)^4 \6x}}\;, 
\end{equation} 
where $\mu_N$ is the Gaussian measure on $\cH_N$ with covariance function 
$[-\Delta + 1]^{-1}$. 

The limit $N\to\infty$ of this sequence of measures is not well-defined, which 
is why a renormalisation procedure is required. The first step of this 
procedure is called Wick renormalisation. It consist in 
replacing~\eqref{eq:phi4_measure} for finite $N$ by 
\begin{equation}
	\label{eq:phi4_measure_Wick} 
	\mu_{\Phi^4_3,N}^{\text{Wick}}(\6\phi) = \frac{1}{Z_N(\eps)}
	\exp\Biggset{-
		\int_\Lambda \Biggpar{\frac12 \norm{\nabla \phi(x)}^2 
			+ \frac{1}{2} \phi(x)^2 + \frac{\eps}{4}\Wick{\phi(x)^4}} \6x}
	\6\phi\;,
\end{equation}
where we write 
\begin{equation}
	\Wick{\phi(x)^n} := H_n\bigpar{\phi(x), C_N^{(1)}}\;.
\end{equation} 
Here $H_n(\cdot, C)$ denotes the $n$th Hermite polynomial with variance $C$, 
and 
\begin{equation}
	C_N^{(1)} := \frac{1}{\abs{\Lambda}} \Tr\bigpar{(-\Delta + 1)^{-1}}
\end{equation} 
is a counterterm which diverges like $N$. 

In the case of the two-dimensional torus, the analogue of the 
measure~\eqref{eq:phi4_measure_Wick} is known to converge to a well-defined 
limit. However, in the three-dimensional case, additional counterterms are 
required. The highly non-trivial result is that three additional such terms are 
sufficient. The correctly renormalised measure takes the form 
(see for instance~\cite[p.~145]{BCGNOPS78})
\begin{multline}
	\label{eq:phi4_measure_BPHZ} 
	\mu_{\Phi^4_3,N}^{\text{BPHZ}}(\6\phi) = \frac{1}{Z_N(\eps)}
	\exp\Biggset{-
		\int_\Lambda \Biggpar{\frac12 \norm{\nabla \phi(x)}^2 \\
			+ \frac{1}{2} \bigbrak{1 - \eps^2 C_N^{(2)}} \phi(x)^2 + 
			\frac{\eps}{4}\Wick{\phi(x)^4}
			+ \eps^2 C_N^{(3)} - \eps^3 C_N^{(4)}} \6x}
	\6\phi\;,
\end{multline}
where the new counterterms are defined as follows. We write $G_N$ for the 
truncated Green function given by 
\begin{equation}
	\label{eq:cov_gff_N} 
	G_N(x,y) = G_N(x-y) 
	:= \sum_{k\in\cK_N} \frac{1}{\lambda_k+1} 
	\underbrace{e_k(x)\overline{e_k(y)}}_{= e_k(x-y)}\;,
\end{equation} 
where $-\lambda_k = (2\pi)^2\norm{k}^2$ are the eigenvalues of the Laplacian on 
$\Lambda$. Then we have 
\begin{align}
	C_N^{(2)} &:= 3! \int_\Lambda G_N(x)^3\6x = \Order{\log N}\;,  \\
	\label{eq:counterterms} 
	C_N^{(3)} &:= \frac{4!}{2!4^2} \int_\Lambda G_N(x)^4\6x
	= \Order{N}\;, \\
	C_N^{(4)} &:= \frac{2^3}{3!4^3} \binom{4}{2}^3 
	\int_{\Lambda}\int_{\Lambda} 
	G_N(x)^2 G_N(y)^2 G_N(x-y)^2 \6x\6y
	= \Order{\log N}\;. 
\end{align}
Our aim in the following is to provide a compact partial proof of this result, 
based on recent developments in combinatorics in the Wiener chaos, on 
Hopf-algebraic methods, and on analytic bounds for BPHZ renormalisation. 
More precisely, we are going to address the question of convergence of the 
truncated partition function $Z_N$ (in the sense of \emph{formal power series})
in the limit $N\to\infty$. 

It will be useful to introduce some additional notation. We 
will use the symbols
\begin{equation}
	\label{def:XY} 
	X = \FGfour := \int_{\T^3} \Wick{\phi(x)^4}\6x\;, \qquad
	Y = \FGtwo := \int_{\T^3} \Wick{\phi(x)^2}\6x
\end{equation} 
for Wick powers, as well as the shorthands 
\begin{equation}
	\alpha := \frac{\eps}{4}\;, \qquad
	\label{eq:alpha_beta_gamma} 
	\beta := \frac12\eps^2 C_N^{(2)}\;, \qquad
	\gamma := \eps^2 C_N^{(3)} - \eps^3 C_N^{(4)}\;. 
\end{equation}
In this way, the ratio of partition functions can be written as 
\begin{equation}
	\frac{Z_N(\eps)}{Z_N(0)}
	= \bigexpec{\e^{-\alpha X - \beta Y - \gamma}}
	= \e^{-\gamma} \bigexpec{\e^{-\alpha X - \beta Y}}\;.
\end{equation} 
Integrals as in~\eqref{eq:counterterms} can be conveniently expressed as 
Feynman diagrams (more precisely, vacuum diagrams). If $\Gamma = (\cV,\cE)$ is 
a multigraph with vertex set $\cV$ and edge set $\cE$ (multiple edges between 
vertices are allowed), then the \emph{valuation} $\Pi_N$ is the map defined by 
\begin{equation}
	\Pi_N(\Gamma) := \int_{\Lambda^\cV} \prod_{e\in\cE} G_N(x_{e_+} - x_{e_-}) 
	\6x\;, 
\end{equation} 
where $e_\pm$ denote the vertices connected by the edge $e$ ($G_N$ being 
even, their order does not matter here). In particular, we have the expressions 
\begin{alignat}{3}
	C_N^{(1)} & = \Pi_N\FGLoop \;, \qquad 
	&& C_N^{(2)} && = 3! \Pi_N\FGIII \;, \qquad \\
	C_N^{(3)} & = \frac{4!}{2!4^2} \Pi_N\FGIV \;, \qquad 
	&& C_N^{(4)} && = \frac{2^3}{3!4^3} \binom{4}{2}^3 \Pi_N\FGVI
\end{alignat}  
for the counterterms. The graphical notation emphasizes how these expressions 
are a consequence of Wick calculus, which states in particular that the 
expectation of a product of Wick powers can be written as a sum over all 
pairings of their \lq\lq legs\rq\rq, also called contractions, see for 
instance~\cite{Peccati_Taqqu_book} as well as Example~\ref{ex:PT} below. 

There are two different questions that one may want to address: 
\begin{enumerate}
	\item Show that the partition function (or its logarithm) admits a perturbative expansion in powers of $\eps$ with coefficients that converge to finite limits as $N \to \infty$.
	
	\item Analyse the Borel summability of the perturbative series. Indeed, it is 
	known that the perturbative expansion of the partition function will remain 
	divergent, even after removing all divergences in terms of $N$. Nonetheless, 
	Borel summation allows to recover information on the partition function from 
	its Borel transform. 
\end{enumerate}
We present a complete answer to the first question in Section~\ref{sec:expansion}.
We address the second question in Section~\ref{sec:Borel} and present some new ideas that help simplify some aspects in the proof of Borel resummability.


\section{Perturbative expansion} 
\label{sec:expansion}


\subsection{Cumulant expansion}
\label{ssec:cumulant}

Define the centred moments 
\begin{align}
	\mu_n &:= (-1)^n \biggexpec{\biggpar{\alpha \FGfour + \beta \FGtwo}^n} \\
	&= (-1)^n \sum_{m=0}^n \binom{n}{m} \alpha^m \beta^{n-m} A_{nm}\;, &
	A_{nm} &:= \biggexpec{\FGfour^m \FGtwo^{n-m}}\;.  
\end{align} 
As already alluded to, the coefficients $A_{nm}$ can be computed using the 
properties of Wick calculus, by summing over all contractions, that is, all 
pairings of legs of different diagrams (see Example~\ref{ex:PT} for more 
details). For instance, we have 
\begin{align}
	\mu_2 &= \alpha^2 4! \Pi_N \FGIV + \beta^2 2! \Pi_N \FGII\;, \\
	\mu_3 &= -\alpha^3 \binom{4}{2}^3 2^3 \Pi_N \FGVI 
	- 3\alpha^2\beta (4^2\cdot 2\cdot 3!) \Pi_N \FGIIIplus \\
	& \quad - 3\alpha\beta^2 4!\, \Pi_N \FGIItwice
	- 8\beta^3 \Pi_N \FGtriangle\;.
\end{align}
We see that each $A_{nm}$ has the form of a combinatorial numerical constant 
times the value of a Feynman diagram, obtained by performing all possible 
contractions. In general, $A_{nm}$ may be a linear combination of Feynman 
diagrams, and these diagrams need not all be connected. For instance, $A_{44}$ 
contains the term
\begin{equation}
	\label{eq:non-connected} 
	3\cdot (4!)^2 \biggpar{\Pi_N \FGIV}^2\;,
\end{equation} 
where the combinatorial factor $3$ counts the number of pairings of the $4$ 
four-vertex diagrams, and each factor $4!$ counts the number of pairwise 
matchings of the legs within each pair. 

The cumulant expansion reads 
\begin{equation}
	\label{eq:cumulant} 
	-\log \expec{\e^{-\alpha X - \beta Y - \gamma}}
	= \gamma - \sum_{n=2}^\infty \frac{\kappa_n}{n!}\;,
\end{equation} 
where the coefficients $\kappa_n$ can be computed recursively 
with the Leonov--Shiraev relation 
\begin{equation}
	\kappa_n = \mu_n - \sum_{m=2}^{n-2} \binom{n-1}{m} \kappa_m \mu_{n-m}\;.
\end{equation} 
It will be useful to write 
\begin{equation}
	\kappa_n = (-1)^n \sum_{m=0}^n \binom{n}{m} \alpha^m \beta^{n-m} B_{nm}\;,
\end{equation} 
where the coefficients $B_{nm}$ are again linear combinations of Feynman 
diagrams. The first few cumulants are 
\begin{align}
	\kappa_2 &= \mu_2\;,\\
	\kappa_3 &= \mu_3\;,\\
	\kappa_4 &= \mu_4 - 3\mu_2^2\;,\\
	\kappa_5 &= \mu_5 - 10\mu_2\mu_3\;.
\end{align}
An important observation is that $\kappa_n - \mu_n$ is always either zero, or a 
sum of products of at least two factors. In terms of Feynman diagrams, this 
means that $\kappa_n - \mu_n$ is a linear combination of non-connected graphs. 

In particular, we see that the term $- 3\mu_2^2$ kills exactly the 
non-connected 
term~\eqref{eq:non-connected} of $\mu_4$. Therefore, $\kappa_4$ is represented 
by a linear combination of \emph{connected} Feynman diagrams. 
The fact that this generalises to all cumulants is well known in the quantum field theory 
literature as the \emph{linked cluster theorem}. 
We refer the reader to the ar\-ti\-cles~\cite[Sec.~3]{Brouder09} and~\cite[Sec.~4.2]{Rivasseau_09} as well as the monograph~\cite[Sec.~2]{Salmhofer_Renormalization}.

\begin{proposition}[Linked cluster theorem] \label{conj:1}
	Every $\kappa_n$ is obtained by projecting~$\mu_n$ onto the space spanned by \emph{connected} Feynman diagrams. 
\end{proposition}
There are many ways in which one can prove the previous proposition. 
We will follow the Peccati--Taqqu approach~\cite{Peccati_Taqqu_book} in which Feynman diagrams naturally arise from so-called \enquote{diagram formulae} that are well-known in Wiener chaos theory.
We believe that this approach is particularly appealing to a probabilistically minded au\-dien\-ce and give more details in the next subsection.


\subsection{A combinatorial \enquote{proof} of Proposition~\ref{conj:1}}
\label{ssec:combinatorial_peccati_taqqu}

With~$X$ and~$Y$ as introduced in~\eqref{def:XY}, we set 
\begin{align}
X & \equiv \FGfour \equiv \int_{\T^3} \Wick{\phi(x)^4}\6x =: \int_{\T^3} X(x) 
\6x\;, \quad X(x) := \FGfourlabel{x}\;,\\
Y & \equiv \FGtwo \equiv \int_{\T^3} \Wick{\phi(x)^2}\6x =: \int_{\T^3} Y(x) 
\6x\;, \quad Y(x) := \FGtwolabel{x}\;, 
\end{align} 
and by~$\mu_N$ we always denote the Gaussian Free Field (GFF) at cut-off 
level~$N$, i.e. the centered Gaussian measure on~$L^2(\Lambda)$ with covariance 
kernel~$G_N(x,y) = G_N(x-y)$ given in~\eqref{eq:cov_gff_N}.

We also define
\begin{equation}
\tilde{G}_N(x) := \sum_{k \in \mathcal{K}_N} \frac{1}{\sqrt{\lambda_k + 1}} 
e_k(x)
\end{equation} 
so that~$(\tilde{G}_N * \tilde{G}_N)(x) = G_N(x)$ when~$*$ denotes 
convolution, as can be verified by a straightforward calculation.

Whenever no explicit measure is mentioned, the reference measure is always that 
of spatial white noise on~$L^2(\Lambda)$, i.e., the centred Gaussian measure 
with covariance given by the Dirac kernel~$\delta(x-y)$.
Accordingly, we write~$\phi \sim \mu$ as  
\begin{equation}
\phi(x) = I_1\bigpar{\tilde{G}_N(x-\cdot)} 
=: \int_\Lambda \tilde{G}_N(x-z) \xi(\6z)\;,
\label{eq:phi}
\end{equation}
where~$I_1$ is the first Wiener-It\^{o} isometry with respect to spatial white 
noise, see for example the textbook by Nua\-lart~\cite{nualart2006malliavin}.
Note that this is consistent with the calculation
\begin{align}
\E^{\mu_N}\biggpar{\phi(x)\phi(y)}
& = 
G_N(x-y) 
=
(\tilde{G}_N * \tilde{G}_N)(x-y)
=
\int_\Lambda 
\tilde{G}_N(z) \tilde{G}_N(z-(x-y)) \6 z \\
& 
=
\int_\Lambda 
\tilde{G}_N(x-z) \tilde{G}_N(y-z) \6 z
=
\E \biggpar{I_1\bigpar{\tilde{G}_N(x-\cdot)} 
I_1\bigpar{\tilde{G}_N(y-\cdot)} }\;,
\end{align}
where we have used that~$G_N$ is even, i.e.~$G_N(x) = G_N(-x)$, and translation 
invariance in the penultimate step.

It is well-known that products of stochastic integrals such as~$\phi$ in 
\eqref{eq:phi} produce correction terms in lower order Wiener-It\^{o} chaoses 
(see e.g.~\cite[Prop.~$1.1.3$]{nualart2006malliavin}) --- but the Wick product 
is the projection onto the highest component, so we have
\begin{equation}
X(x) = \Wick{\phi(x)^4} = I_4(\tilde{G}_N(x-\cdot)^{\otimes 4})\;, \quad
Y(x) = \Wick{\phi(x)^2} = I_2(\tilde{G}_N(x-\cdot)^{\otimes 2})\;.
\end{equation}
Recall that 
\begin{equation}
\kappa(X_1,\ldots,X_n) 
:= \left.\frac{\partial^n}{\partial t_1 \dots \partial t_n} \log 
\E\biggpar{\exp\biggpar{\sum_{\ell=1}^n t_\ell X_\ell}}\right|_{t_1 = \ldots 
t_n 
= 0}
\end{equation}
denotes the~\emph{cumulant functional}. 
With~$\kappa_n(X) = \kappa(X, \ldots, X)$ where~$\kappa$ has $n$ entries, we 
will use the well-known binomial-type formula 
\begin{align}
\kappa_n(\alpha X + \beta Y)
& = 
\sum_{k=0}^n \binom{n}{k} \alpha^k \beta^{n-k} 
\kappa(\underbrace{X,\ldots,X}_{k~\text{times}},\underbrace{Y,\ldots,Y}_{
(n-k)~\text{times}})\;.
\end{align}
Multi-linearity of~$\kappa$ also gives
\begin{equation}
\kappa(\underbrace{X,\ldots,X}_{k~\text{times}},\underbrace{Y,\ldots,Y}_{
(n-k)~\text{times}})
=
\int_{\Lambda^n} 
\kappa\bigpar{X(x_1),\ldots,X(x_k),Y(x_{k+1}),\ldots,Y(x_n)}\; 
\6x_{1:n}
\end{equation}
where~$\6 x_{1:n} := \6 x_1 \ldots \6 x_n$.
The following theorem implies the validity of Proposition~\ref{conj:1}.

\begin{theorem} \label{thm:cumu_graphs}
The identity
\begin{equation}
\kappa(\underbrace{X,\ldots,X}_{k~\text{times}},\underbrace{Y,\ldots,Y}_{
	(n-k)~\text{times}})
=
\sum_{\Gamma \in \cG} \Pi_N \Gamma
\end{equation}
holds, where~$\cG = \cG(k,n)$ denotes the set of connected 
multigraphs without self-loops that correspond to pairwise matchings.
\end{theorem}

We will not introduce all the terminology in the previous statement abstractly, 
but rather illustrate it in a specific case.
As the reader will see, the arguments easily generalise to all combinations 
of~$k$ and~$n$ and thus lead to a \enquote{proof by example.}
\begin{remark}
Essentially, the previous theorem is a direct consequence 
of~\cite[Coro.~$7.3.1$]{Peccati_Taqqu_book}, a ge\-ne\-ra\-li\-sa\-tion of 
Wick's theorem followed by a projection onto~\enquote{connected diagrams} to 
account for the cumulant.
Since the whole book~\cite{Peccati_Taqqu_book} is written in the language of 
set-partition combinatorics, we hope our example aids the reader in seeing the 
connections clearly.
\end{remark}

\begin{example}[$k=2, n=3$: \enquote{Proof} of Theorem~\ref{thm:cumu_graphs}]
\label{ex:PT} 
We consider
\begin{align}
X(x_1) & = \FGfourlabel{x_1} =   
I_4(\underbrace{\tilde{G}_N(x_1-\cdot)^{\otimes 
		4}}_{=: f_1(x_1;\;\cdot)}) \quad \rightarrowtriangle \quad n_1 = 4\;, \\
X(x_2) & = \FGfourlabel{x_2} =  I_4(\underbrace{\tilde{G}_N(x_2-\cdot)^{\otimes 
		4}}_{=: f_2(x_2;\;\cdot)}) \quad \rightarrowtriangle \quad n_2 = 4\;, \\
Y(x_3) & = \FGtwolabel{x_3} =  I_2(\underbrace{\tilde{G}_N(x_3-\cdot)^{\otimes 
		2}}_{=: f_3(x_3;\;\cdot)}) \quad \rightarrowtriangle \quad n_3 = 2\;.
\end{align}
\begin{enumerate}[label=(\arabic*)]
\item In a first step, we convert the legs of the diagrams~$X(x_i)$, $i=1,2$, 
and~$Y(x_3)$ into nodes and keep the 
label~$x_i$ on the left side of the row. 
Accordingly, we have $n= \sum_{i=1}^3 n_i = 10$ nodes in the following diagram: 
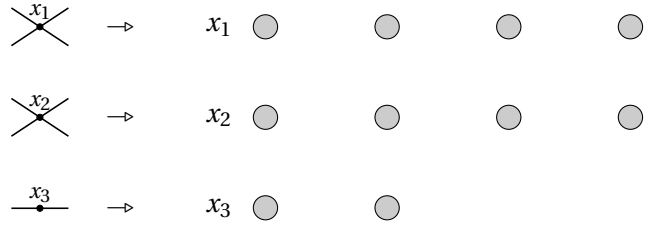
\begin{center}
	\begin{tikzpicture}  
		[scale=.8,auto=center,dot/.style={circle,draw=black,fill=black!20}] 
		
		\node (g1) at (-3,6) {$\FGfourlabel{x_1} \quad \rightarrowtriangle \quad$};  
		\node (g2) at (-3,4.5) {$\FGfourlabel{x_2} \quad \rightarrowtriangle \quad$};  
		\node (g3) at (-3,3) {$\FGtwolabel{x_3} \quad \rightarrowtriangle \quad$};  
		\node[dot, label={[label distance=1.5mm] west:$x_1$}] (a1) at (0,6) {};  
		\node[dot] (a2) at (2,6)  {}; 
		\node[dot] (a3) at (4,6)  {};  
		\node[dot] (a4) at (6,6) {};  
		\node[dot, label={[label distance=1.5mm] west:$x_2$}] (a5) at (0,4.5)  {};  
		\node[dot] (a6) at (2,4.5)  {};  
		\node[dot] (a7) at (4,4.5)  {};  
		\node[dot] (a8) at (6,4.5)  {};  
		\node[dot, label={[label distance=1.5mm]  west:$x_3$}] (a9) at (0,3)  {};  
		\node[dot] (a10) at (2,3)  {};  
		
	\end{tikzpicture} 
\end{center}
\item We form pairwise matchings of these nodes, signified by 
lines between two nodes, abi\-ding by the following rules:
\begin{enumerate}[label=(\roman*)]
	\item One must not match two nodes that are in the same 
	row. This would correspond to self-loops in the associated graphs. Peccati and 
	Taqqu call these matchings \enquote{non-flat.}
	\item The resulting matchings must be such that one cannot 
	divide the rows without intersecting one line that symbolises a matching of two 
	nodes. 
	Otherwise, one could partition the rows into two or more 
	subsets of rows and form pairwise matchings within each subset. This would 
	correspond to disconnected graphs.
\end{enumerate}
The set that contains all of these matchings is 
called~$\mathcal{M}_2([n],\pi^\star)$ by Peccati and Taqqu, see 
point~\ref{it:part_pistar} below for the definition of~$\pi^\star$ in our 
context. 

We denote the specific matching in the following diagram by~$\sigma$: 
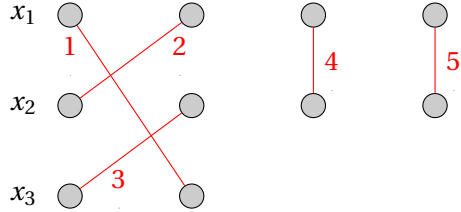
\begin{center}
	\begin{tikzpicture}  
		[scale=.8,auto=center,dot/.style={circle,draw=black,fill=black!20}, 
		zero/.style={circle,radius=0pt,draw=black,fill=black!100, inner sep=0pt, outer 
			sep=0pt, minimum size=.0cm}]] 
		forms 
		
		\node[dot, label={[label distance=1.5mm] west:$x_1$}] (a1) at (0,6) {};  
		\node[dot] (a2) at (2,6)  {}; 
		\node[dot] (a3) at (4,6)  {};  
		\node[dot] (a4) at (6,6) {};  
		\node[dot, label={[label distance=1.5mm] west:$x_2$}] (a5) at (0,4.5)  {};  
		\node[dot] (a6) at (2,4.5)  {};  
		\node[dot] (a7) at (4,4.5)  {};  
		\node[dot] (a8) at (6,4.5)  {};  
		\node[dot, label={[label distance=1.5mm]  west:$x_3$}] (a9) at (0,3)  {};  
		\node[dot] (a10) at (2,3)  {};

		\draw[red] (a1) -- (a10);
		\draw[red] (a2) -- (a5);  
		\draw[red] (a3) -- (a7);  
		\draw[red] (a4) -- (a8);  
		\draw[red] (a6) -- (a9);  
		
		
		\node[zero, label={[label distance=1.5mm] \textcolor{red}{$1$}}] (1) at (0,5) 
		{};
		\node[zero, label={[label distance=1.5mm] \textcolor{red}{$2$}}] (2) at (1.8,5) 
		{};
		\node[zero, label={[label distance=1.5mm] \textcolor{red}{$3$}}] (3) at 
		(0.8,2.8) {};
		\node[zero, label={[label distance=1.5mm] \textcolor{red}{$4$}}] (4) at 
		(4.3,4.75) {};
		\node[zero, label={[label distance=1.5mm] \textcolor{red}{$5$}}] (5) at 
		(6.3,4.75) {};
		
	\end{tikzpicture} 
\end{center}
\item One converts back these matchings into pairings of the 
legs of the variables~$X(x_1)$, $X(x_2)$, and~$Y(x_3)$. For the specific 
matching $\sigma$ above, this leads to the following multigraph (or 
\emph{vacuum Feynman diagram}): 
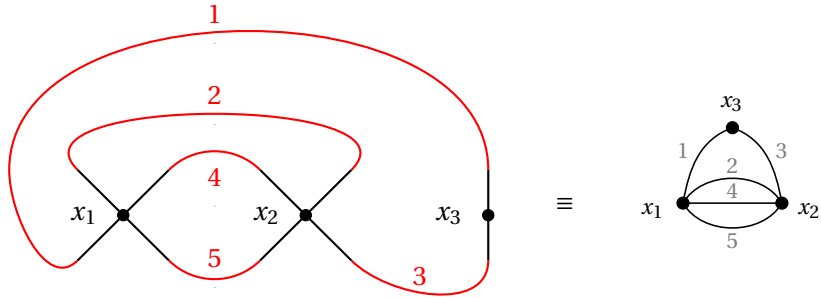
\begin{equation}
	\raisebox{-12mm}{
		\begin{tikzpicture}  
			[scale=.6,auto=center,dot/.style={circle,draw=black,fill=black!100, inner 
				sep=1.5pt, outer sep=0pt}, 
			zero/.style={circle,radius=0pt,draw=black,fill=black!100, inner sep=0pt, outer 
				sep=0pt, minimum size=.0cm}] 
			
			\path[use as bounding box] (-2.6,1.1) rectangle (8.1,7.6);

			\node[dot, label={[label distance=1.5mm] west:$x_1$}] (a1) at (0,3) {};  
			\node[dot, label={[label distance=1.5mm] west:$x_2$}] (a2) at (4,3)  {}; 
			\node[dot, label={[label distance=1.5mm] west:$x_3$}] (a3) at (8,3)  {};  
			\node[zero] (1ol) at (-1,4) {};
			\node[zero] (1or) at (1,4) {};
			\node[zero] (1ul) at (-1,2) {};
			\node[zero] (1ur) at (1,2) {};

			\draw[thick] (1ol) -- (a1);
			\draw[thick] (1or) -- (a1);
			\draw[thick] (1ul) -- (a1);
			\draw[thick] (1ur) -- (a1);
			\node[zero] (2ol) at (3,4) {};
			\node[zero] (2or) at (5,4) {};
			\node[zero] (2ul) at (3,2) {};
			\node[zero] (2ur) at (5,2) {};
			\draw[thick] (2ol) -- (a2);
			\draw[thick] (2or) -- (a2);
			\draw[thick] (2ul) -- (a2);
			\draw[thick] (2ur) -- (a2);
			\node[zero] (3o) at (8,4) {};
			\node[zero] (3u) at (8,2) {};
			\draw[thick] (3u) -- (a3);
			\draw[thick] (3o) -- (a3);
			
			
			\draw[red, thick] (1or) to[out=45,in=135] (2ol);
			\draw[red, thick] (1ur) to[out=-45,in=-135] (2ul);
			\draw[red, thick] (2ur) to[out=-45,in=-90] (3u);
			\draw[red, thick] (1ol) to[out=135,in=45] (2or);
			\draw[red, thick] (1ul) to[out=225,in=-90] (-2.5,4) to[out=90,in=90] (3o); 
			
			
			\node[zero, label={[label distance=1.5mm] \textcolor{red}{$4$}}] (4) at (2,3.2) 
			{};
			\node[zero, label={[label distance=1.5mm] \textcolor{red}{$5$}}] (5) at (2,1.4) 
			{};
			\node[zero, label={[label distance=1.5mm] \textcolor{red}{$2$}}] (2) at (2,5) 
			{};
			\node[zero, label={[label distance=1.5mm] \textcolor{red}{$1$}}] (1) at (2,6.8) 
			{};
			\node[zero, label={[label distance=1.5mm] \textcolor{red}{$3$}}] (3) at (6.5,1) 
			{};
			
			\drawbox;
		\end{tikzpicture} 
	}
	\qquad\equiv\qquad
	\raisebox{-3mm}{
		\begin{tikzpicture}[scale=2]
			\path[use as bounding box] (-0.3, -0.2) rectangle (0.95,0.8);
			\draw[semithick] (0,0)--(0.65,0);
			\draw[semithick] (0,0) edge [out=60,in=120] (0.65,0);
			\draw[semithick] (0,0) edge [out=-60,in=-120] (0.65,0);
			\draw[semithick] (0,0) edge [out=80,in=-160] (0.325,0.5);
			\draw[semithick] (0.65,0) edge [out=100,in=-20] (0.325,0.5);
			\TreeVertex{0}{0};
			\TreeVertex{0.65}{0};
			\TreeVertex{0.325}{0.5};
			\node[] at (-0.2,-0.05) {\footnotesize{$x_1$}};
			\node[] at (0.83,-0.05) {\footnotesize{$x_2$}};
			\node[] at (0.32,0.66) {\footnotesize{$x_3$}};
			\node[] at (0,0.35) {\footnotesize{$\textcolor{gray}{1}$}};
			\node[] at (0.325,0.25) {\footnotesize{$\textcolor{gray}{2}$}};
			\node[] at (0.65,0.35) {\footnotesize{$\textcolor{gray}{3}$}};
			\node[] at (0.325,0.08) {\footnotesize{$\textcolor{gray}{4}$}};
			\node[] at (0.325,-0.25) {\footnotesize{$\textcolor{gray}{5}$}};
			\drawbox;
		\end{tikzpicture}
	}
\end{equation}
\item \label{it:part_pistar}There is an analytical expression corresponding to 
this multi-graph respectively the mat\-ching that leads to it. Let us 
illustrate how this can be obtained using the result of Peccati and 
Taqqu~\cite[Coro.~$7.3.1$]{Peccati_Taqqu_book}.

\begin{enumerate}[label=(4.\arabic*)]
	\item Note that
	\begin{equation}
		\biggpar{\bigotimes_{k=1}^3 f_k(x_k,\cdot)}(z_{1:n})
		=
		f_1(x_1,z_{1:4}) f_2(x_2,z_{5:8}) f_3(x_3,z_{9:10})\;.
		\label{eq:fct_tp}
	\end{equation}
	The partition~$\pi^\star$ in Peccati's and Taqqu's book corresponds exactly to 
	the partitioning of the indices of the~$z_i$'s:
	\begin{align}
		\pi^\star  
		& = \bigset{\{1,\ldots,n_1\}, \{n_1+1,\ldots,n_1+n_2\}, 
			\{n_1+n_2+1,\ldots,n_1+n_2+n_3\}} \\
		& = \bigset{\{1,\ldots,4\}, \{5,\ldots,8\},\{9,10\}}\;.
	\end{align}
	\item 
	The function~$f_{\sigma,\ell} \equiv f_{\sigma,3}$
	built from~\eqref{eq:fct_tp} is obtained by simply identifying 
	the variables that are matched via~$\sigma$.
	In precise terms, this means that
	\begin{align}
		f_{\sigma,3}(z_{1:5})
		{}={}& 
		\tilde{G}_N(x_1-z_1)\tilde{G}_N(x_1-z_2)
		\tilde{G}_N(x_1-z_4) \tilde{G}_N(x_1-z_5) \\
		{}\times{}& \tilde{G}_N(x_2-z_2)\tilde{G}_N(x_2-z_3)
		\tilde{G}_N(x_2-z_4)\tilde{G}_N(x_2-z_5) \\
		{}\times{}& \tilde{G}_N(x_3-z_1)\tilde{G}_N(x_3-z_3)\;.
	\end{align} 
	\item We then find
	\begin{equation}
		\int_{\Lambda^5} f_{\sigma,3}(z_{1:5}) \6 z_{1:5}
		=
		G_N(x_2-x_1)^3 G_N(x_3-x_1) G_N(x_3 - x_2)
		=
		\FGIIIpluslabel{x_1}{x_2}{x_3}
	\end{equation}
	by repeatedly using the identity $(\tilde{G}_N * \tilde{G}_N)(x) 
	= G_N(x)$.
	\item Integrating over all the~$x_i$'s (above called the 
	\emph{valuation}~$\Pi_N$) we then find
	\begin{align}
		& \thinspace
		\int_{\Lambda^3} \int_{\Lambda^5} f_{\sigma,3}(z_{1:5}) 
		\6 z_{1:5} \6 x_{1:3} \\
		= & \
		\int_{\Lambda^3} G_N(x_2-x_1)^3 G_N(x_3-x_1) G_N(x_3 - 
		x_2)  \6 x_{1:3} \\
		= & \
		\int_{\Lambda^{\abs{\cV}}} \prod_{e \in \cE} G_N(x_{e_+} 
		- x_{e_-})  \6 x
		=
		\Pi_N \Gamma
	\end{align}
	for
	\begin{equation}
		\Gamma
		=
		\FGIIIplus\;.
	\end{equation}
	In the literature, the previous drawing sometimes sym\-bo\-li\-ses~$\Gamma \in 
	\cG$ and sometimes denotes~$\Pi_N \Gamma \in \R$. 
	We follow the former convention.
\end{enumerate}
\end{enumerate}
The above procedure clearly generalises to arbitrary values of~$n$ and~$k$ 
and to different matchings that produce~$\Gamma \in \cG = \cG(k,n)$. 
Therefore, Theorem~\ref{thm:cumu_graphs} follows by the same route.
\end{example}


\subsection{BPHZ renormalisation}
\label{ssec:BPHZ}

We now examine the cumulant expansion~\eqref{eq:cumulant} in more detail.
Each coefficient $A_{nm}$ and $B_{nm}$ is a linear combination of Feynman 
diagrams $\Gamma^{(k)}_{nm}$ having each 
\begin{itemize}
\item $m$ vertices of degree $4$,
\item $n-m$ vertices of degree $2$, 
\item $n+m$ edges.
\end{itemize}
We associate with a multigraph $\Gamma = (\cV, \cE)$ a degree given by 
\begin{equation}
\deg(\Gamma) = 3(\abs{\cV}-1) - \abs{\cE}\;,
\end{equation} 
so that $\deg(\Gamma^{(k)}_{nm}) = 2n - m - 3$ for all $k$. We call a 
diagram $\Gamma$ \emph{divergent} if $\deg\Gamma\leqs0$. 

For small divergent diagrams, one can check that their value diverges 
like $N^{-\deg(\Gamma)}$, possibly with logarithmic corrections. This is 
however not true for many larger diagrams, because of the presence of divergent 
subdiagrams. In fact, there is only one possible divergent subdiagram in our 
situation, namely the \lq\lq bubble\rq\rq\ 
\begin{equation}
\FGIII\;,
\end{equation}
the value of which diverges like $\log N$. 

BPHZ renormalisation, named after Bogoliubov, Parasiuk, Hepp and 
Zimmermann~\cite{Bogoliubow_Parasiuk,Hepp66,Zimmermann69} provides a way of 
dealing with these divergent subdiagrams. It can be formulated in a convenient 
way by using the Connes--Kreimer extraction-contraction coproduct on graphs
\cite{Connes_Kreimer_2000a,Connes_Kreimer_2001}, given by
\begin{equation}
\label{eq:ec_coproduct} 
\Delta(\Gamma) =  \Gamma\otimes\unit + \unit\otimes\Gamma 
+ \sum_{\unit\neq\overline\Gamma\subsetneq\Gamma} 
\overline\Gamma\otimes(\Gamma/\overline\Gamma)\;, 
\end{equation} 
where $\unit$ denotes the empty graph, the sum ranges over all divergent 
subdiagrams, and $\Gamma/\overline\Gamma$ is obtained by contracting all edges 
in $\overline\Gamma$ to one vertex. We further define a (twisted) antipode 
$\cA$ 
as the linear map satisfying $\cA(\unit) = \unit$, and extended inductively by 
\begin{align}
\cA(\Gamma) 
&= -\cM(\cA\otimes\id)(\Delta\Gamma - \Gamma\otimes\unit) \\
&= -\Gamma - \sum_{\unit\neq\overline\Gamma\subsetneq\Gamma} 
\cA(\overline\Gamma)\cdot(\Gamma/\overline\Gamma)\;, 
\label{eq:antipode} 
\end{align} 
where the product $\cdot$ denotes the disjoint union of graphs, and 
$\cM(\Gamma_1\otimes\Gamma_2) = \Gamma_1\cdot\Gamma_2$. 
For instance, 
\begin{align}
\Delta\Bigpar{\FGIIIplus} &= \FGIIIplus\otimes\unit 
+ \unit\otimes\FGIIIplus + \FGIII \otimes \FGII\;,\\
\cA\Bigpar{\FGIIIplus} &= - \FGIIIplus + \FGIII \cdot \FGII\;.
\label{eq:antipode_FGIIIplus} 
\end{align}
A \emph{character} is by definition a linear form $g:\cG\to\R$, 
$\Gamma\mapsto\pscal{g}{\Gamma}$ which is also multiplicative, in the sense 
that 
\begin{equation}
\pscal{g}{\Gamma_1\cdot\Gamma_2} = 
\pscal{g}{\Gamma_1}\pscal{g}{\Gamma_2}
\qquad \forall\Gamma_1, \Gamma_2\in\cG\;.
\end{equation} 
To such a character, we can associate a renormalisation transformation, given 
by the linear map $M^g: \cG\to\cG$ defined by  
\begin{equation}
M^g(\Gamma) := (g\otimes\id)\Delta\Gamma\;.
\end{equation} 
The \emph{BPHZ character} is given by 
\begin{equation}
\pscal{g^{\text{BPHZ}}}{\Gamma} := 
\begin{cases}
\Pi_N\cA(\Gamma) & \text{if $\deg\Gamma\leqs0$\;,} \\
0 & \text{otherwise\;.}
\end{cases}
\end{equation}
We then define 
\begin{align}
\label{eq:PiNBPHZ} 
\PiNBPHZ(\Gamma) 
&= \Pi_N M^{g^{\text{BPHZ}}} (\Gamma) \\
&= (g^{\text{BPHZ}}\otimes\Pi_N) \Delta\Gamma\;. 
\end{align} 
A compact way of writing this is to introduce the map $\tilde\cA$ defined by 
\begin{equation} \label{eq:tildeA}
\tilde\cA(\Gamma) := \cA(\Gamma) 1_{\deg\Gamma\leqs0}\;,
\end{equation} 
which implies 
\begin{align}
\label{eq:PiNBPHZ} 
\PiNBPHZ(\Gamma) 
&= (\Pi_N\otimes\Pi_N)(\tilde\cA\otimes\id) \Delta\Gamma \\
&= (\Pi_N\tilde\cA\otimes\Pi_N)\Delta\Gamma\;. 
\end{align} 
The following commutative diagram summarises the situation: 
\begin{equation}
\begin{tikzcd}[column sep=large, row sep=large]
\displaystyle \cG
\arrow[r, "\Pi_N"]
\arrow[d, "(\tilde\cA\otimes\id)\Delta"'] 
\arrow[rd, "\PiNBPHZ"]
\arrow[dd, out=180, in=180, distance=50pt, 
" M^{g^{{\text{BPHZ}}}} = (\Pi_N\tilde\cA\otimes\id)\Delta"']
& \R  \\
\displaystyle \cG\otimes\cG
\arrow[r, "\quad\Pi_N\otimes\Pi_N\qquad"] 
\arrow[d, "\Pi_N\otimes\id"'] 
& \R 
\\
\cG
\arrow[ru, "\Pi_N"']
\end{tikzcd}
\end{equation}

The interest of this construction is that one can show that $\PiNBPHZ(\Gamma)$ 
is bounded uniformly in $N$ if $\Gamma$ is non-divergent, and otherwise 
diverges 
like $N^{-\deg(\Gamma)}$, possibly with logarithmic 
corrections~\cite{Hairer_BPHZ,Berglund_Bruned_19}. 

The aim of the remainder of this section is to give a mostly algebraic proof of 
the following combinatorial result. 

\begin{theorem} \label{thm:cumu_expansion}
\label{thm:combinatorial} 
For $0 \leqs m \leqs n$, write 
\begin{equation}
B_{nm} = \sum_k b^{(k)}_{nm} \Pi_N\bigpar{\Gamma^{(k)}_{nm}}\;, 
\end{equation} 
where each sum runs over finitely many $k$, the $b^{(k)}_{nm}$ are 
combinatorial coefficients, and each $\Gamma^{(k)}_{nm}$ is a connected Feynman 
diagram with $m$ vertices of degree $4$, $n-m$ vertices of degree $2$, and 
$n+m$ edges. Then 
\begin{equation}
\label{eq:conj2} 
\sum_{n=2}^\infty \frac{\kappa_n}{n!}
= -\sum_{p=2}^\infty \frac{1}{p!} (-\alpha)^p 
\sum_k b^{(k)}_{pp} \PiNBPHZ\bigpar{\Gamma^{(k)}_{pp}}\;,
\end{equation} 
where equality is in the sense of formal power series. 
\end{theorem}

\begin{corollary}
\label{cor:combinatorial}  
All terms in the perturbative cumulant expansion~\eqref{eq:cumulant} are bounded 
uniformly in the cut-off $N$. 
\end{corollary}
\begin{proof}
Since $\deg(\Gamma^{(k)}_{pp}) = p-3$, the only divergent 
$\Gamma^{(k)}_{pp}$ are those with $p\in\set{2,3}$. The choice of $\gamma$ 
in~\eqref{eq:alpha_beta_gamma} precisely compensates these two terms. The 
result follows at once. 
\end{proof}

\subsection{Zimmermann's forest formula}
\label{ssec:Zimmermann} 

A first step in proving Theorem~\ref{thm:combinatorial} is to obtain a simpler 
expression than~\eqref{eq:antipode} for the twisted antipode. This is provided 
by Zimmermann's forest formula~\cite{Zimmermann69,Hairer_BPHZ}, which reads 
\begin{equation}
\cA(\Gamma) = - \sum_{\cF} (-1)^{\abs{\cF}} \cC_\cF\Gamma\;.
\end{equation} 
Here the sum ranges over all forests $\cF$ not containing $\Gamma$, where a 
\emph{forest} is a set of subgraphs of $\Gamma$ which are pairwise either 
included in one another, or vertex-disjoint. The operator $\cC_\cF$ extracts 
all subgraphs in $\cF$ from $\Gamma$ (for a forest, this operation is 
independent of the order of the elements of $\cF$). 

In our simple situation, forests are just unions of disjoint bubbles. The 
forest formula thus takes the following form: if $\Gamma$ contains $g$ bubbles, 
then 
\begin{equation}
\cA(\Gamma) = - \sum_{S\subset\set{1,\dots,g}} 
(-1)^{\abs{S}} \FGIII^{\abs{S}} \cC_S \Gamma\;,
\end{equation} 
where we write $\cC_S$ for the operation consisting in contracting all bubbles 
labelled by an element of $S$ (for an arbitrary fixed labelling of the 
bubbles). 

As a result, using the fact that 
\begin{equation}
\label{eq:bubble}
\Pi_N \Bigpar{\FGIII} = \frac{\beta}{3\eps^2}
= \frac{\beta}{48\alpha^2}\;,
\end{equation} 
we obtain 
\begin{equation}
\PiNBPHZ\bigpar{\Gamma^{(k)}_{pp}}
= - \sum_{S\subset\set{1,\dots,g}} 
\Biggpar{-\frac{\beta}{48\alpha^2}}^{\abs{S}} 
\Pi_N\bigpar{\cC_S \Gamma^{(k)}_{pp}}\;,
\end{equation} 
where $\cC_S \Gamma^{(k)}_{pp}$ is a diagram with $p-\abs{S}$ vertices, 
$\abs{S}$ of which are of degree $2$.


\subsection{A Hopf-algebra-flavoured proof}
\label{ssec:Hopf}

As we have already alluded to in the last sections, it is well-known that the 
triple~$(\cG,\cdot,\Delta)$ can canonically be turned into a Hopf algebra, 
itself isomorphic to the \emph{Connes--Kreimer Hopf 
algebra}~\cite{ConnesKreimer98} of rooted trees.
Therefore, it is natural to expect that a particularly elegant proof of 
Theorem~\ref{thm:combinatorial} may be achieved if one were to interpret $X = 
\FGfour$ and~$Y = \FGtwo$ as \emph{monomials} in a polynomial Hopf algebra.
Indeed, such a construction has been performed by Ebrahimi-Fard et 
al.~\cite{EFPTZ18}: In this section, we describe how to adapt it for our 
purposes.\footnote{Note that~$X$ and~$Y$ do satisfy the assumptions in that 
article: They have moments of all orders because probabilistic~$L^p$-norms 
coincide in all homogeneous Wiener-It\^{o} 
chaoses~\cite[Thm.~$3.50$]{Janson_book_08}.}

We let~$\X := (X,Y)$ and for each~$\bn =(n_1,n_2) \in \N^2$ then 
define 
\begin{equation}
\X^{\bn} := X^{n_1}Y^{n_2}\;, \quad 
\X^{(0,0)} := \bs{1}\;, \quad
H := \operatorname{span}\bigset{\X^{\bn}: \ \bn \in \N^2}\;.
\end{equation}

\begin{remark}
In~\cite{EFPTZ18}, the authors allow for~$\X = (X_a: \ a \in 
I)$ where~$I$ is some (possibly infinite) index set and then consider the set
\begin{equation}
M(I) = \bigset{B: I \to \N: \ B(a) \neq 0 \ \text{for 
		finitely many} \ a \in \cA}\;.
\end{equation}   
In our setting~$\abs{I} = 2$, so we can identify~$M(I) \simeq \N^2$, 
i.e. we just consider~$B = \bn \in \N^2$.
In particular, we still consider polynomials~$\X^{\bn}$ 
as above and, contrary to~\cite{EFPTZ18}, do not base our analysis on the 
(isomorphic) vector space freely generated from~$M(I)$.

We repeat some of the constructions in~\cite{EFPTZ18}, albeit in a way slightly 
adapted to our setting. For~$\bn,\bm, \bk \in \N_0^2$, we define
\begin{align} 
\cM(\X^{\bn} \otimes \X^{\bm})
& := \X^{\bn} \cdot \X^{\bm} 
:= \X^{\bn \cdot \bm}, \quad \bk \cdot \bm 
:= (k_1 + m_1, k_2 + m_2),
\\
\hat{\Delta} \X^{\bn}  & :=
\sum_{\substack{\bk,\bm \in \N_0^2 \\ \bk \cdot 
		\bm = \bn}} \binom{\bn}{\bm,\bk} \X^{\bk} \otimes \X^{\bm}, 
\quad \binom{\bn}{\bm,\bk} := \prod_{i=1}^2 \frac{n_i!}{m_i! k_i!}\;,
\end{align}
with a slight abuse of notation for the product~$\cdot$\thinspace. 
It is proved in~\cite{EFPTZ18} that~$\Delta$ defines a coproduct 
on~$H$. 
\end{remark}
In the previous sections, we have seen that the (twisted) 
antipode~$\cA$~(resp.~$\tilde{\cA}$) plays a crucial role in the renormalisation 
of multigraphs~$\Gamma \in \cG$.
In this section, we want to define a corresponding map~$\hat{\cA}_\eta$ acting 
on~$H$, as well as a map~$\cP$ that sends~$H$ to~$\cG$, such that the two diagrams in figure~\ref{fig:commutative_diagrams} between spaces respectively objects commute.

\begin{landscape}
\begin{figure}[h]
\centering
\begin{tikzcd}[column sep=large, row sep=large]
	H 
	\arrow[r, "\cP"] 
	\arrow[d, "\chi_\eta = 
	(\hat\cA_\eta\otimes\id)\hat\Delta"']
	& \displaystyle \cG
	\arrow[r, "\Pi_N"]
	\arrow[d, "(\tilde\cA\otimes\id)\Delta"'] 
	\arrow[rd, "\PiNBPHZ"]
	& \R  \\
	H\otimes H
	\arrow[d, "\cM"'] 
	& \displaystyle \cG\otimes\cG
	\arrow[r, "\quad\Pi_N\otimes\Pi_N\qquad"] 
	\arrow[d, "\Pi_N \otimes \id"'] 
	& \R 
	\\
	H
	\arrow[r, "\cP"] 
	& \cG
	\arrow[ru, "\Pi_N"']
\end{tikzcd}

\vspace{3em}

\begin{tikzcd}[cramped,column sep=large, row sep=large]
	\e^{-\alpha X} 
	\arrow[r, "\cP", mapsto] 
	\arrow[d, "\chi_\eta = 
	(\hat\cA_\eta\otimes\id)\hat\Delta"', 
	mapsto]
	& \displaystyle \sum_{n,k} \tfrac{(-\alpha)^n}{n!} b_{nn}^{(k)} 
	\Gamma_{nn}^{(k)} 
	\arrow[r, "\Pi_N", mapsto]
	\arrow[d, "(\tilde\cA\otimes\id)\Delta"', mapsto] 
	\arrow[rd, "\PiNBPHZ", mapsto]
	& \displaystyle \sum_{n} \tfrac{(-\alpha)^n}{n!} B_{nn}
	\\
	\displaystyle
	\sum_{n,m} \tfrac{(-\alpha)^n(-\eta)^m}{(n-2m)!m!} Y^m\otimes X^{n-2m}
	\arrow[d, "\cM"', mapsto] 
	& \displaystyle \sum_{n,k,S} 
	\tfrac{(-\alpha)^n}{n!} b_{nn}^{(k)}
	\tilde\cA\bigpar{\raisebox{2mm}{\scalebox{0.7}{\FGIII}}^{\abs{S}}} 
	\otimes \cC_S \Gamma_{nn}^{(k)}
	\arrow[r, "\Pi_N\otimes\Pi_N", mapsto] 
	\arrow[d, "\Pi_N \otimes \id"', mapsto] 
	& \displaystyle \sum_{n} \tfrac{(-\alpha)^{n-m}(-\beta)^m}{(n-m)!m!} B_{nm} 
	\\
	\e^{-\alpha X - \beta Y}
	\arrow[r, "\cP", mapsto] 
	& \displaystyle \sum_{n,m,k} 
	\tfrac{(-\alpha)^{n-m}(-\beta)^m}{m!(n-m)!} 
	b_{nm}^{(k)} 
	\Gamma_{nm}^{(k)} 
	\arrow[ru, "\Pi_N"', mapsto]
\end{tikzcd}
\vspace{1em}
\caption{Commutative diagrams between spaces (top) and objects (bottom).}
\label{fig:commutative_diagrams}
\end{figure}
\end{landscape}

The following definition introduces the desired maps~$\hat{\cA}_\eta$ and~$\cP$.
Recall that
\begin{equation} \label{eq:double_factorial}
(2\ell-1)!! 
:= \prod_{i=1}^{\ell} (2i-1)
= \frac{(2\ell)!}{2^\ell \ell!}\;.
\end{equation} 

\begin{definition} \label{def:antipode_2}
We define the linear map~$\hat{\cA}_\eta: H \to H$ by 
\begin{equation}
\label{eq:antipode2}
\hat{\cA}_\eta \X^{\bn} := 
\begin{cases}
	(2\ell-1)!! (-2 \eta Y)^{\ell} & \quad 
	\text{if} \quad n_1 = 2\ell, \ n_2 = 0\;,\\
	0 & \quad \text{otherwise}, \\
\end{cases}
\end{equation}
Note that the condition on the first line is satisfied precisely if \;$\X^{\bn} 
= X^{2\ell} \ \text{for some} \ \ell \in \N_0$ and, in particular, implies 
that~$\hat{\cA}_\eta \mathbf{1} = \mathbf{1}$.  
We also decree that~$\hat{\cA}_\eta$ is linear w.r.t. infinite sums and define
\begin{equation}
\chi_\eta: H \to H \otimes H\;, \quad 
\chi_\eta(\X^{\bn}) := (\hat{\cA}_\eta \otimes 
\id) \hat{\Delta} \X^{\bn} \;,
\end{equation}
as well as the operation~$\cP: H \to \cG$ by
\begin{equation}
\cP \X^{\bn} := \sum_{\Gamma \in \cG(n_1,n_1+n_2)} \Gamma\;,
\end{equation}
for\; $\cG(n_1,n_1+n_2)$ as given in Theorem~\ref{thm:cumu_graphs}. 
Note that this corresponds exactly to~$\X^{\bn} = X^{n_1}Y^{n_2}$. In other 
words,~$\cP$ coincides with~$\kappa$ up to taking the valuation~$\Pi_N$.
\end{definition}

\begin{remark}
The previous definition might seem somewhat \enquote{ad-hoc}, so let us 
explain the motivation behind it. 
Essentially, we want to mirror the down-facing arrows in the middle 
of the two diagrams contained in Figure~\ref{fig:commutative_diagrams}, respectively.
We neglect combinatorial factors which are accounted for by the 
binomial coefficient. 
\begin{enumerate}[label=(\arabic*)]
\item 
We know that the extraction-contraction co-product~$\Delta$ 
given in \eqref{eq:ec_coproduct} extracts divergent subgraphs of some 
\emph{connected} graph~$\Gamma$. 

The only divergent diagrams in our setting are~$\FGIV$, 
$\FGVI$, and the bubble~$\FGIII$ --- but all the vertices of the first two 
graphs have valence~$4$, so they cannot be subgraphs of~$\Gamma$  because it is 
connected, see also Section~\ref{ssec:Zimmermann}. 

\item The bubble, however, is produced by pairing three of the 
legs of~$X = \FGfour$ with three legs of another instance of~$X$; extracting a 
bubble then locally produces~$\FGIII \cdot \FGtwo$, i.e.~$\eta Y$ for~$\eta 
\simeq \nicefrac{\beta}{\alpha^2}$.
Because~\enquote{two instances of~$X$ produce~$Y$}, the latter 
should count double for the argument to be correct; that is the reason for the 
action of~$\hat{\cA}_\eta$ on even powers of~$X$.

\item Encoding the factor~$\eta$ in~$\hat{\cA}_\eta$ (and not~$\hat{\Delta}$) 
seemed more convenient in computations and resembles the action of~$\cA$ 
in~\eqref{eq:antipode_FGIIIplus} more closely.

\item The factor $(2\ell-1)!!$ counts all possible ways to pair the $2\ell$ 
four-vertex diagrams.

\item The map~$\hat{\cA}_\eta$ should act trivially on non-divergent diagrams, 
so it sends all monomials that are not even powers of~$X$ to~$0$.  
\qedhere
\end{enumerate}
\end{remark}

The following proposition proves that the two leftmost down-facing arrows in 
the diagrams contained in Figure~\ref{fig:commutative_diagrams} correspond to 
well-defined operations:

\begin{proposition}
\label{prop:exponential} 
The identity~$(\cM \circ \chi_\eta) \e^{-{\alpha X}} = \e^{-\alpha X - 
\beta Y}$ holds in the sense of formal power series.
\end{proposition}

\begin{proof}
Observe that~$X^n = \X^{\bn}$ for~$\bn = (n_1,n_2) = (n,0)$. By 
definition, we thus have
\begin{equation}
(\cM \circ \chi_\eta)X^n = \sum_{\substack{\bm,\bk \in \N^2 \\ \bm \cdot \bk = 
		\bn}} \binom{\bn}{\bm,\bk} \hat{\cA}_\eta (\X^{\bm}) 
\X^{\bk}\;,
\label{eq:cMchi_Xn}
\end{equation}
and since~$n_2 = 0$, the condition~$\bm \cdot \bk = \bn$ implies that~$m_2 + 
k_2 
= 0$. As~$m_2, k_2 \in \N_0$, we then find~$m_2 = k_2 = 0$ and the preceding 
formula simplifies to
\begin{equation}
(\cM \circ \chi_\eta)X^n 
= \sum_{\substack{m,k \in \N \\ m + k = n}} \frac{n!}{m! k!} \hat{\cA}_\eta 
(X^{m}) 
X^{k}\;
= \sum_{m=0}^n \frac{n!}{m! (n-m)!} \hat{\cA}_\eta (X^{m}) 
X^{n-m}.
\label{eq:cMchi_Xn_2}
\end{equation}
Accounting for the definition of the antipode in~\eqref{eq:antipode2}, we 
arrive 
at the identity
\begin{align}
(\cM \circ \chi_\eta)X^n 
& = \sum_{\ell=0}^{\lfloor \nicefrac{n}{2} \rfloor} \frac{n!}{(2\ell)! 
	(n-2\ell)!} \hat{\cA}_\eta (X^{2\ell}) 
X^{n-2\ell} \\
& \label{eq:M_comp_chi} = \sum_{\ell=0}^{\lfloor \nicefrac{n}{2} \rfloor} \frac{n!}{(2\ell)! 
	(n-2\ell)!} (2\ell-1)!! (-2\eta Y)^{\ell}
X^{n-2\ell} \label{eq:cMchi_Xn_3} \\
& = \sum_{\ell=0}^{\lfloor \nicefrac{n}{2} \rfloor} \frac{n!}{\ell! (n-2\ell)!} 
(-\eta Y)^{\ell}
X^{n-2\ell}.
\end{align}
Next, we expand the exponential function
\begin{equation}
\e^{-\alpha X} = \sum_{n=0}^{\infty} \frac{(-1)^n}{n!} \alpha^n 
X^n\;
\end{equation}
so that the previous computation and~$\eta = 
\nicefrac{\beta}{\alpha^2}$ implies that 
\begin{equation}
(\cM \circ \chi_\eta) \e^{-{\alpha X}} 
=
\sum_{n=0}^{\infty} \frac{(-1)^n}{n!} 
\sum_{\ell = 0}^{\lfloor \nicefrac{n}{2} \rfloor} \frac{n!}{(n-2\ell)! 
	\ell!} (- \beta Y)^\ell (\alpha X)^{n-2\ell}\;.
\end{equation}
We now set~$p := n-\ell$ and~$q := n-2\ell$ so that 
\begin{equation}
\ell = p-q\;, \quad n = 2p - q\;, \quad (-1)^n  = (-1)^{n-2\ell} (-1)^{2\ell} 
= (-1)^q \quad \text{for} \quad \ell = 0,\ldots,\lfloor \nicefrac{n}{2} 
\rfloor\;, 
\end{equation}
and then reorganise the series accordingly to get
\begin{align}
(\cM \circ \chi_\eta) \e^{-{\alpha X}} 
& =
\sum_{p=0}^{\infty} 
\sum_{q = 0}^{p} \frac{1}{q! (p-q)!} (- \beta Y)^{p-q} (-\alpha 
X)^q \\
& = 
\sum_{p=0}^{\infty} \frac{1}{p!}
\sum_{q = 0}^{p} \binom{p}{q} (- \beta Y)^{p-q} (-\alpha X)^q 
= \e^{-\alpha X + \beta Y}\;.
\end{align}
The proof is complete.
\end{proof}

\begin{remark}
If one \emph{formally} identifies
\begin{equation}
	X \; \longleftrightarrow \; Z \sim \cN(0,1), \quad  2 \eta Y\; 
	\longleftrightarrow 
	\; E[Z^2] = 1,
\end{equation}
it is interesting to observe that the RHS of~\eqref{eq:M_comp_chi} reads
\begin{equation}
	\sum_{\ell=0}^{\lfloor \nicefrac{n}{2} \rfloor} \binom{n}{2\ell} (2\ell-1)!! (-1)^\ell Z^{n-2\ell} = H_n(Z,1)
\end{equation}
where~$H_n(\cdot,1)$ is the $n$-th Hermite polynomial defined w.r.t. the Gaussian measure of variance~$1$ in the convention of~\cite[Remark~on~p.~146]{Peccati_Taqqu_book}.
In other words: We see that the map~$\cM \circ \chi_\eta$ \emph{formally} behaves like a Wick product when acting on polynomials in~$X$. 
We leave further investigations of this observation for future work.
\end{remark}

It remains to check that the diagrams are indeed commutative. This is the 
content of the next section.

\subsection{Combinatoric proof of the diagram's commutativity}
\label{ssec:combinatorics}

To complete the proof of Theorem~\ref{thm:combinatorial}, we need to 
show that the identity
\begin{align}
\label{eq:commutative} 
\cP \circ \cM \circ \chi_\eta
&= (\Pi_N\otimes\id)\circ(\tilde\cA\otimes\id)\Delta\circ\cP \\
&= (\Pi_N\tilde\cA \otimes \id)\Delta \circ \cP
\end{align} 
holds on the space spanned by all monomials $X^n$. This is equivalent to 
showing that the following diagram commutes:
\begin{equation}  \label{comdiag_Xn_small}
\begin{tikzcd}[column sep=large, row sep=large]
X^n 
\arrow[r, "\cP", mapsto] 
\arrow[d, "\cM \circ\chi_\eta"', 
mapsto]
& \displaystyle \sum_{k}  b_{nn}^{(k)} 
\Gamma_{nn}^{(k)} 
\arrow[d, "(\Pi_N\otimes\id)(\tilde\cA\otimes\id)\Delta", mapsto] 
\\
\displaystyle
\sum_{n,m} \tfrac{n!(-\eta)^m}{(n-2m)!m!} Y^m X^{n-2m}
\arrow[r, "\cP", mapsto] 
& \displaystyle \sum_{k,S} b_{nn}^{(k)} \Biggpar{\frac{-\eta}{48}}^{\abs{S}}
\cC_S \Gamma_{nm}^{(k)} 
\end{tikzcd}
\end{equation}
We have already obtained in the proof of Proposition~\ref{prop:exponential} the 
expression 
\begin{equation}
(\cM \circ \chi_\eta)(X^n)
= \sum_{m=0}^{\lfloor n/2 \rfloor} 
\frac{n!}{(n-2m)!m!}(-\eta)^m Y^m X^{n-2m}\;.
\end{equation} 
On the other hand, recalling that $\eta = \nicefrac{\beta}{\alpha^2}$, we get 
\begin{equation}
(\Pi_N\tilde\cA \otimes \id)\Delta\Gamma
= -\sum_{m=0}^{\lfloor n/2 \rfloor} 
\Biggpar{\frac{\eta}{48}}^m 
\sum_{S\colon \abs{S} = m} \cC_S\Gamma\;,
\end{equation} 
where the second sum runs over all sets of $m$ bubbles (if $\Gamma$ has fewer 
than $m$ bubbles, the last sum is zero by definition). Comparing the last two 
expressions, we see that~\eqref{eq:commutative} holds if
\begin{equation}
\label{eq:combi_nm} 
\cP(Y^m X^{n-2m}) 
= \frac{(n-2m)!m!}{48^m n!}
\sum_{S\colon \abs{S} = m} \cC_S(\cP X^n) 
\end{equation} 
for any $n\geqs1$. Recall that both sides of this relation involve in general a 
sum over several Feynman diagrams. There is, however, a natural identification 
between these diagrams, so that we may lighten the notation by pretending that 
there is only one diagram in each sum. 

We now observe that $\cP(Y^m X^{n-2m})$ is a (sum of) diagram(s) having $m$ 
vertices of degree $2$, $n-m$ vertices of degree $4$, and $2n-m$ edges. To 
produce the corresponding term on the right-hand side, we argue that instead of 
extracting bubbles on the right-hand side, we can also insert bubbles on the 
left-hand side. This amounts to inserting $2m$ four-vertex diagrams at the 
vertices of degree $2$ of $\cP(Y^m X^{n-2m})$. To do that, there are 
\begin{itemize}
\item $\binom{n}{2m}$ ways of selecting the $2m$ four-vertex diagrams; 
\item $(2m-1)!!$ ways to pair the $2m$ four-vertex diagrams;
\item $4^2\cdot3! = 96$ ways of matching six pairs of legs of each set of two 
four-vertex diagrams, amounting to $96^m$ matchings;
\item and finally, $m!$ ways of inserting the resulting bubbles at the $m$ 
vertices of index $2$.
\end{itemize}
Multiplying all the above combinatorial factors, by~\eqref{eq:double_factorial} 
we arrive at 
\begin{equation}
\frac{48^m n!}{m!(n-2m)!}
\end{equation} 
ways of inserting the four-vertex diagrams. This is indeed compatible with the 
desired relation~\eqref{eq:combi_nm}. The proof of 
Theorem~\ref{thm:combinatorial} is complete.

We close this section with an example that deals with mixed monomials.
\begin{example}
As we have seen, it suffices for the purposes of this article that the diagram 
in~\eqref{comdiag_Xn_small} commutes for the base monomial~$X^n$.
However, the following example gives us hope that the commutativity might still 
be true for monomials containing non-trivial powers of~$Y$.

Let~$\X^{\bn} = X^2 Y$, i.e.~$\bn = (n_1,n_2) = (2,1)$.
We find
\begin{align}
\hat{\Delta} (X^2 Y) 
& = X^2 Y \otimes \mathbf{1} + \mathbf{1} \otimes X^2 Y + X^2 \otimes Y + Y 
\otimes X^2 + 2  XY \otimes X + 2 X \otimes XY\;, \\
\chi_\eta (X^2 Y) 
& =  \mathbf{1} \otimes X^2 Y + \hat{\cA}_\eta(X^2) \otimes Y\;,
\end{align}
and then
\begin{equation}
\label{eq:computation_mixed_monomial}
(\cM \circ \chi_\eta) (X^2 Y)
=
X^2 Y - 2 \eta Y^2\;.
\end{equation}
Next, note that
\begin{equation}
\label{eq:PXXY} 
\cP(X^2 Y) = 4^2\cdot 2!\cdot 3! \,\FGIIIplus\;, 
\qquad \cP(Y^2) = 2! \,\FGII\;,
\end{equation}
where the combinatorial factor in the first expression counts the $4^2$ ways of 
choosing one leg in each four-vertex diagram, the $2!$ ways of matching these 
with the legs of the two-vertex diagram, and the $3!$ pairwise matchings of the 
remaining legs. 
From~the equalities in~\eqref{eq:antipode_FGIIIplus}, \eqref{eq:tildeA}, as 
well 
as \eqref{eq:bubble}, we obtain 
\begin{equation}
(\Pi_N\otimes\id)(\tilde\cA\otimes\id)\Delta \FGIIIplus
=
\FGIIIplus - \Pi_N \FGIII \cdot \FGII
=
\FGIIIplus - \frac{\eta}{48} \FGII\;.
\end{equation}
Applying the expressions~\eqref{eq:PXXY} of $\cP$ to $X^2Y$ and $X^2Y-2\eta 
Y^2$, we see that the commutativity relation~\eqref{eq:commutative} is indeed 
satisfied. We expect that the same conclusion holds for other mixed monomials, 
modulo a suitable encoding of the combinatorics in the definition of the 
map~$\hat{\cA}_\eta$ from Definition~\ref{def:antipode_2}.
\end{example}


\section{Borel resummation} 
\label{sec:Borel}


In this section, we examine the question of whether the perturbative
expansion~\eqref{eq:cumulant}, though not convergent, can nevertheless be 
related to a convergent quantity. A positive answer to this question can be 
given thanks to the theory of Borel summation. This fact has been known in 
quantum field theory for quite a while, though first proofs of Borel 
summability 
for the $\Phi^4_2$ and $\Phi^4_3$ 
model~\cite{Eckmann_Magnen_Seneor_75,Magnen_Seneor_77} were quite difficult. 
In fact, to the best of our knowledge, this question is still open in the case of~$\Phi^4_3$ in infinite volume.

Here we show that modern analytical tools, which combine Hopf-algebraic 
methods and a decomposition originally obtained by Hepp in~\cite{Hepp66}, allow for a clean formulation and simplified proof of the \emph{remainder bounds}, the first from a set of two sufficient conditions in an improved Borel summation result by Sokal~\cite{Sokal80}. 
However, as already mentioned in the introduction, our proof requires a relatively strong moment bound, see Remark~\ref{rmk:conditions_sokal} below. 
The second condition for Sokal's result, \emph{local analyticity}, is assumed true to begin with.
%
Our arguments for the remainder bounds are strongly based on the presentation in~\cite{Hairer_BPHZ} which was made more quantitative in~\cite{Berglund_Bruned_19}. 
We start in Section~\ref{sec:watson_sokal} by presenting the main ideas of Borel summation, in particular Sokal's result, and then apply it in Section~\ref{ssec:zerodim} to the zero-dimensional $\Phi^4$ model, whose partition function is simply an integral over $\R$.
See also~\cite{Rivasseau_09} for a more detailed account of various resummation 
techniques for that model. Then we show in Section~\ref{ssec:dimthree} how 
\emph{some} of these ideas can be extended to the three-dimensional case, using in particular 
methods introduced by Hepp.

\subsection{Watson's and Sokal's theorems} \label{sec:watson_sokal}

Certain divergent series can be resummed by a procedure known as Borel 
summation. Consider a formal power series 
\begin{equation}
A(\eps) = \sum_{n\geqs0} a_n \eps^n\;.
\end{equation} 
We can rewrite it as 
\begin{equation}
A(\eps) = \sum_{n\geqs0} a_n\eps^n \frac{\Gamma(n+1)}{n!}  
= \sum_{n\geqs0} \frac{a_n\eps^n}{n!} \int_0^\infty t^n \e^{-t} \6t\;.
\end{equation} 
Define the Borel-transformed power series by interchanging the sum and the 
integral, that is 
\begin{equation}
A_{\text{Borel}}(\eps) 
:= \int_0^\infty \e^{-t} \sum_{n\geqs0} \frac{a_n\eps^nt^n}{n!} \6t 
= \int_0^\infty \e^{-t} \cB A(\eps t) \6t\;, 
\end{equation} 
which is the Laplace transform of the Borel sum 
\begin{equation}
\cB A(t) := \sum_{n\geqs0} \frac{a_n}{n!}t^n\;.
\end{equation} 

Watson's theorem~\cite{Watson_1912} gives conditions under which 
$A_{\text{Borel}}(\eps)$ admits the asymptotic series $A(\eps)$. 
In particular, $A$ should be analytic in a sector 
$|\arg\eps| < \pi/2 + \delta$, $\abs{\eps} < R$ for strictly positive $\delta, R$.
In~\cite{Sokal80}, Sokal has proved the following improvement of 
Watson's theorem. 

\begin{theorem}[Sokal] \label{thm:sokal}
Let $A$ be analytic on the disk~$D_R = 
\setsuch{\eps}{\re\eps^{-1} > R^{-1}}$ for some~$R > 0$.
Assume further that $A$ admits the 
asymptotic expansion 
\begin{equation}
\label{eq:asympt_Sokal} 
A(\eps) = \sum_{k=0}^{n-1} a_k \eps^k + R_n(\eps)\;,
\end{equation} 
where
\begin{equation}
\label{eq:bound_Sokal} 
\bigabs{R_n(\eps)} 
\leqs C r^n n! \abs{\eps}^n
\end{equation} 
for some $C, r>0$, uniformly in $n$ and $\eps$ in $D_R$. Then $\cB A(t)$ 
converges for $\abs{t} < 1/r$, and has an analytic continuation to a 
$1/r$-neighbourhood of the positive real axis. Furthermore, $A$ can be 
represented by the absolutely convergent integral 
\begin{equation}
A(\eps) = \frac{1}{\eps} \int_0^\infty \e^{-t/\eps} \cB A(t) \6t
\end{equation} 
for any $\eps\in D_R$. 
\end{theorem}

	\begin{figure}[h]
		\centering
		\begin{tikzpicture}  
		[scale=0.95,>=stealth',auto=center,dot/.style={circle,draw=black,fill=black!20},main node/.style={draw,circle,fill=white,minimum
			size=3pt,inner sep=0pt}] 
		
		\def\xWatson{-0.5}
		
		\path[fill=violet!25] (0,0) -- (\xWatson,-2) -- (4,-2) -- (4,2) -- (\xWatson,2) -- cycle;
		
		\draw[thick,violet] (\xWatson,-2) -- (0,0) -- (\xWatson,2);
		
		\draw[thick,blue,fill=blue!25] (1.5,0) circle (1.5);
		
		\draw[semithick,->] (-2,0) -- (5.0,0);
		\draw[semithick,->] (0,-2.4) -- (0,2.7);
		
		\node[main node,blue,semithick,fill=white] at (1.5,0) {};
		
		\node[] at (4.4,0.2) {$\re\eps$};
		\node[] at (0.4,2.3) {$\im\eps$};
		\node[] at (1.5,-0.4) {$\frac{R}{2}$};
		
		\node[purple!30!black] at (3.2,1.7) {Watson};
		\node[blue] at (1.5,0.5) {Sokal};
		
		\end{tikzpicture} 
		\caption{A graphical representation of the analyticity requirements in Watson's and Sokal's theorems.}
	\end{figure}
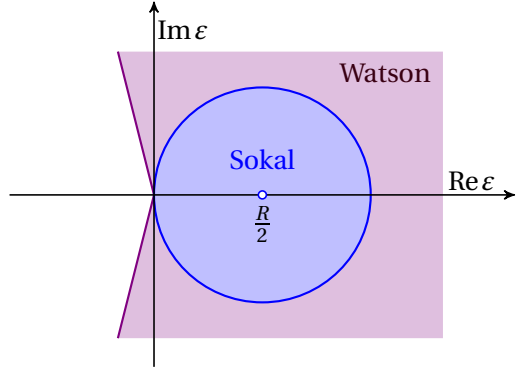

\subsection{The case of the $\Phi^4_0$ model} 
\label{ssec:zerodim} 

The $\Phi^4_0$ model is simply the $\Phi^4$ model for a field $\phi$ defined at 
a single point. Its potential is 
\begin{equation}
V(\phi) = \frac12\phi^2 + \frac{\eps}{4}\phi^4\;, 
\end{equation} 
and its partition function is given by 
\begin{equation}
Z(\eps) = \int_{-\infty}^\infty \e^{-V(\phi)} \6\phi
= \int_{-\infty}^\infty \e^{-\phi^2/2} \e^{-\eps\phi^4/4} \6\phi\;.
\end{equation} 
The integral is clearly well-defined for $\eps \geqs 0$. It can also be 
extended to complex values of $\eps$, at least if $\re\eps\geqs0$, and possibly 
to other complex values. However, the integral is clearly not convergent for 
real $\eps < 0$. Therefore, $Z$ is not analytic in a neighbourhood of $\eps = 
0$, and does not admit a convergent expansion in powers of $\eps$. 

If we nevertheless expand the exponential, we obtain 
\begin{equation}
Z(\eps) 
\asymp \sum_{n\geqs0} \frac{1}{n!} \biggpar{-\frac{\eps}{4}}^n 
\int_{-\infty}^\infty \phi^{4n} \e^{-\phi^2/2} \6\phi \;,
\end{equation}
where the symbol $\asymp$ denotes an asymptotic expansion. We can interpret 
$\e^{-\phi^2/2}$ as the density of a Gaussian measure (up to normalisation), 
which yields 
\begin{align}
Z(\eps) 
&\asymp \sqrt{2\pi}\sum_{n\geqs0} \frac{1}{n!} 
\biggpar{-\frac{\eps}{4}}^n 
\bigexpecin{\mu}{\phi^{4n}} \\
&= \sqrt{2\pi}\sum_{n\geqs0} \biggpar{-\frac{\eps}{4}}^n 
\frac{(4n-1)!!}{n!} \;,
\end{align}
where $(4n-1)!!$ is defined in~\eqref{eq:double_factorial}, and we have used 
the Isserlis--Wick theorem to compute the moments of the normal law. Recalling 
that $n! = \Gamma(n+1)$ and using~\eqref{eq:double_factorial}, we get 
\begin{equation}
Z(\eps) \asymp 
\sqrt{2\pi}\sum_{n\geqs0} \biggpar{-\frac{\eps}{16}}^n 
\frac{\Gamma(4n+1)}{\Gamma(2n+1)\Gamma(n+1)}\;.
\end{equation} 
The general term of this formal series can be analysed by using Legendre's 
duplication formula for the Gamma function 
\begin{equation}
\Gamma(2z) = \frac{1}{\sqrt{\pi}} 2^{2z-1} 
\Gamma(z)\Gamma(z+\tfrac12)\;,
\end{equation} 
which yields 
\begin{equation}
\Gamma(4n+1) = 4n \Gamma(4n) 
= \frac{2}{\sqrt{\pi}} n 2^{4n} \Gamma(2n)\Gamma(2n+\tfrac12)\;.
\end{equation} 
Therefore 
\begin{align}
Z(\eps) 
&\asymp 2 \sqrt{2} \sum_{n\geqs0} (-\eps)^n 
\frac{n\Gamma(2n)\Gamma(2n+\tfrac12)}{\Gamma(n+1)\Gamma(2n+1)} \\
\label{eq:Zeps0}
&= \sqrt{2} \sum_{n\geqs0} (-\eps)^n 
\frac{\Gamma(2n+\tfrac12)}{\Gamma(n+1)} \\
&= \frac1{\sqrt{\pi}} \sum_{n\geqs0} (-4\eps)^n 
\frac{\Gamma(n+\tfrac14)\Gamma(n+\tfrac34)}{\Gamma(n+1)}\;,
\label{eq:Zeps}
\end{align}
where we used $2n\Gamma(2n) = \Gamma(2n+1)$ to get the second line. 
Stirling's formula implies 
\begin{equation}
\Gamma(z+\alpha) = z^\alpha \Gamma(z) \biggbrak{1 + 
\biggOrder{\frac{\alpha-1}{z}}}\;,
\end{equation} 
showing that the general term in the series~\eqref{eq:Zeps} diverges like 
$\Gamma(n)$. Therefore, the series is indeed divergent. 

\begin{remark}
A more direct way of obtaining the asymptotic expansion~\eqref{eq:Zeps0} is to 
notice that 
\begin{equation}
\label{eq:Z_Laplace} 
Z(\eps) 
= 2 \int_0^\infty \e^{-\phi^2/2} \e^{-\eps\phi^4/4} \6\phi
= \sqrt{2}\int_0^\infty \e^{-t} \frac{\e^{-\eps t^2}}{\sqrt{t}} 
\6t\;,
\end{equation} 
where we have used the change of variables $\phi = \sqrt{2t}$. Expanding 
the exponential, we obtain 
\begin{equation}
Z(\eps) 
\asymp \sqrt{2}\sum_{n\geqs0} \frac{(-\eps)^n}{n!}\int_0^\infty 
t^{2n-\frac12} 
\e^{-t}\6t
= \sqrt{2}\sum_{n\geqs0} (-\eps)^n \frac{\Gamma(2n+\frac12)}{n!}
\end{equation} 
which agrees with~\eqref{eq:Zeps0}. 
\end{remark}

Applying the Borel transform to the expansion~\eqref{eq:Zeps} we find 
\begin{equation}
\label{eq:BZt} 
\cB Z(t) = \frac{1}{\sqrt{\pi}} \sum_{n\geqs0} b_n t^n\;,
\end{equation} 
where
\begin{equation}
b_n = (-4)^n \frac{\Gamma(n+\tfrac14)\Gamma(n+\tfrac34)}{\Gamma(n+1)^2}
= \frac{(-4)^n}{n} \biggpar{1 + \biggOrder{\frac1n}}\;.
\end{equation} 
The series~\eqref{eq:BZt} has radius of convergence $\frac14$, with a pole 
located at $-\frac14$. One can thus expect that it admits an analytic 
continuation to a domain including all positive reals, so that its Laplace 
transform indeed converges. 

To apply Sokal's Theorem~\ref{thm:sokal}, we write 
\begin{equation}
\frac{Z(\eps)}{\sqrt{2\pi}} 
= \sum_{k=0}^{n-1} \frac{1}{k!} 
\biggpar{-\frac{\eps}{4}}^n \bigexpec{\phi^{4k}} 
+ \Biggexpec{\e^{-\eps\phi^4/4} - \sum_{k=0}^{n-1} \frac{1}{k!} 
\biggpar{-\frac{\eps}{4}}^k \phi^{4k}}\;.
\end{equation} 
We then use the fact that for any $n\in\N$, one has the Taylor expansion
\begin{align}
\e^{-z} 
&= \sum_{k=0}^{n-1} \frac{(-z)^k}{k!}
+ (-z)^n \int_0^1 \int_0^{t_1} \dots \int_0^{t_{n-1}} 
\e^{-t_n z} \6t_n\dots\6t_1 \\
&= \sum_{k=0}^{n-1} \frac{(-z)^k}{k!}
+ (-z)^n \int_0^1 \int_0^1 \dots \int_0^1 s_1^{n-1} s_2^{n-2} \ldots s_{n-1}
\e^{-s_1\dots s_n z} \6s_n\ldots\6s_1\;,
\label{eq:Taylor_remainder} 
\end{align} 
showing that for any $z$ with positive real part, 
\begin{equation}
\Biggabs{\e^{-z} - \sum_{k=0}^{n-1} \frac{1}{k!}(-z)^k} 
\leqs \frac{1}{n!} \abs{z}^n\;.
\end{equation} 
This implies that $Z(\eps)/\sqrt{2\pi}$ satisfies~\eqref{eq:asympt_Sokal} with 
a remainder $R_n$ such that  
\begin{equation}
\bigabs{R_n(\eps)} 
\leqs \frac{1}{n!} \Biggpar{\frac{\abs{\eps}}{4}}^n 
\bigexpec{\phi^{4n}}\;.
\end{equation} 
By the above computations, the remainder indeed meets Sokal's conditions.


\subsection{The case of the $\Phi^4_3$ model} 
\label{ssec:dimthree} 

As mentioned above, Borel summability of the perturbation expansions of 
correlation functions (or Schwinger functions) of the $\Phi^4_d$ model has been 
proved in~\cite{Eckmann_Magnen_Seneor_75} in the case $d=2$, both in finite and infinite volume, and 
in~\cite{Magnen_Seneor_77}  in the finite volume case when $d=3$. The proofs are based on cluster 
expansion techniques from statistical physics and are quite technical. 

Here we outline a comparatively clean proof of the \emph{remainder bounds}---albeit under a relatively strong moment bound assumption.
Together with local analyticity of the partition function (which, again, we do not prove, but assume) this implies Borel summability of its expansion via Sokal's theorem.
The essential and non-trivial analytical ingredient for our proof is a bound on the value of 
BPHZ-renormalised Feynman diagrams, explained in~\cite{Hairer_BPHZ}, and made 
more quantitative in~\cite{Berglund_Bruned_19}. 

The main result in this subsection is the following:

\begin{proposition} \label{prop:borel_dimthree}
Let~$R > 0$ and assume that the following two statements are true.
\begin{enumerate}[label=(\roman*)]
	\item \label{prop:borel_dimthree:i} The map
	\begin{equation}
		\eps \mapsto 
		\frac{Z_N(\eps)}{Z_N(0)} 
		= \bigexpec{\e^{-\alpha X - \beta Y - \gamma}}
	\end{equation}	
	is analytic on~$D_R$. 
	\item \label{prop:borel_dimthree:ii} The following bound holds:
	\begin{equation}
		\bigabs{(\PiNBPHZ\circ\cP)(X^n \e^{-\theta \alpha X})} \lesssim (2n)! \quad \text{for all} \quad n \geq 4, \ \theta \in [0,1]\;.
		\label{prop:borel_dimthree_moment_bound}
	\end{equation} 
\end{enumerate}
Then, the function
$\eps \mapsto \log \frac{Z_N(\eps)}{Z_N(0)}$
is Borel summable.
\end{proposition}

\begin{remark} \label{rmk:conditions_sokal}
We emphasise that \emph{both} the statements in~\ref{prop:borel_dimthree:i} and~\ref{prop:borel_dimthree:ii} in the previous proposition are non-trivial assumptions that require considerable effort (and perhaps a different set of tools) to be checked.
However, although we have not been able to prove that, we believe that our framework should allow us to obtain a bound of type~\eqref{prop:borel_dimthree_moment_bound} with an unspecified constant~$C(n)$.
This would allow us to strengthen the statement of our main result, Theorem~\ref{thm:cumu_expansion}, to say that the \enquote{lo\-ga\-rith\-mic partition function of the~$\Phi^4_3$ theory admits an \emph{asymptotic expansion}} and to prove that claim within our framework.
\end{remark}

Recall from Figure~\ref{fig:commutative_diagrams} that we have the following commutative diagram:

\begin{equation} \label{eq:recap_comm_diagram}
\begin{tikzcd}[column sep=large, row sep=large]
\e^{-\alpha X}
\arrow[d, "\cM \circ \chi_\eta"'] 
\arrow[rd, "\PiNBPHZ \circ \cP"] \\
\e^{ -\alpha X - \beta Y}
\arrow[r, "\Pi_N \circ \cP"] 
& \log \E[\e^{ -\alpha X - \beta Y}]
\end{tikzcd}
\end{equation}

The following observation allows us to subtract divergent terms:
\begin{lemma} \label{lem:subtraction_divergent}
For 
\begin{equation}
F(X) := \sum_{p=4}^\infty \frac{(-\alpha)^p}{p!} X^p
\end{equation}
we have
\begin{equation}
\log \bigexpec{\e^{-\alpha X - \beta Y - \gamma}} = (\PiNBPHZ\circ\cP)F(X\;).
\end{equation}
As a consequence, the following diagram commutes as well: 
\begin{equation} 
\begin{tikzcd}[column sep=large, row sep=large]
	F(X)
	\arrow[d, "\cM \circ \chi_\eta"'] 
	\arrow[rd, "\PiNBPHZ \circ \cP"] \\
	\e^{ -\alpha X - \beta Y - \gamma}
	\arrow[r, "\Pi_N \circ \cP"] 
	& \log \E[\e^{ -\alpha X - \beta Y - \gamma}] \equiv \log \frac{Z_N(\eps)}{Z_N(0)}
\end{tikzcd}
\end{equation}
\end{lemma}

\begin{proof}
We write 
\begin{equation}
\e^{\alpha X} = P(X) + F(X)\;, \quad
P(X) 
= \sum_{p=0}^{3} \frac{(-\alpha)^p}{p!} X^p\;.
\end{equation}
Then, 
\begin{equation}
\log \bigexpec{\e^{-\alpha X - \beta Y - \gamma}} 
=
- \gamma + \log \bigexpec{\e^{-\alpha X - \beta Y}} 
=
- \gamma + (\PiNBPHZ\circ\cP)\e^{-\alpha X}
\end{equation}
where the last equality is true by commutativity of the diagram in~\eqref{eq:recap_comm_diagram}. 
Furthermore, by linearity we have
\begin{equation}
(\PiNBPHZ\circ\cP) P(X)
=
\sum_{p=0}^{3} \frac{(-\alpha)^p}{p!} (\PiNBPHZ\circ\cP)X^p
\end{equation}
where
\begin{equation}
(\PiNBPHZ\circ\cP) \unit = 0 = (\PiNBPHZ\circ\cP) X 
\end{equation}
and
\begin{equation}
(\PiNBPHZ\circ\cP) X^2 = 4! \Pi_N \FGIV\;, \quad
(\PiNBPHZ\circ\cP) X^3 = 2^3 \binom{4}{2}^3 \Pi_N \FGVI\;.
\end{equation}
Recalling that~$\alpha = \eps/4$, we thus find that 
\begin{align}
(\PiNBPHZ\circ\cP) P(X)
& =
\frac{\eps^2}{4^2 2!} 4! \Pi_N \FGIV
+ 
\frac{\eps^3}{4^3 3!} 2^3 \binom{4}{2}^3 \Pi_N \FGVI \\
& = 
\eps^2 C_N^{(3)} + \eps^3 C_N^{(4)}
=
\gamma\;.
\end{align}
Note that the last equality is just the definition of~$\gamma$, see~\eqref{eq:alpha_beta_gamma}.
In summary, we have
\begin{equation}
\log \bigexpec{\e^{-\alpha X - \beta Y - \gamma}} 
=
(\PiNBPHZ\circ\cP)[\e^{-\alpha X} - P(X)]
=
(\PiNBPHZ\circ\cP)F(X).
\end{equation}
The addendum follows immediately from the commutativity of the diagram in~\eqref{eq:recap_comm_diagram}.
\end{proof}

In order to analyse Borel summability, we decompose
\begin{equation}
F(X) = \sum_{p=4}^{n-1} \frac{(-\alpha)^p}{p!} X^p + \sum_{p=n}^{\infty} \frac{(-\alpha)^p}{p!} X^p =: S_n + R_n
\end{equation}
and then get the following result for~$S_n$:

\begin{lemma} \label{lem:asymp_Sn}
We have
\begin{equation}
(\PiNBPHZ\circ\cP)S_n \asymp \sum_{p=4}^{n-1} 
\Order{p!} 
\eps^p\;.
\end{equation}
\end{lemma}

\begin{proof}
Recall that
\begin{equation}
(\PiNBPHZ\circ\cP) X^p = \sum_{k} b_{pp}^{(k)} \, \PiNBPHZ (\Gamma_{pp}^{(k)})\;.
\end{equation}
where the terms~$b_{pp}^{(k)}$ and~$\Gamma_{pp}^{(k)}$ were introduced in Theorem~\ref{thm:cumu_expansion} above.
A straightforward extension of~\cite[Prop.~6.1]{Berglund_Bruned_19} 
shows that whenever $\Gamma$ has strictly positive degree, one has 
\begin{equation}
\abs{\PiNBPHZ \Gamma}
\leqs K^{\abs{\cE(\Gamma)}}
\end{equation} 
where $\abs{\cE(\Gamma)}$ is the number of edges of $\Gamma$, and $K$ is a 
constant depending only on the Green function $G = \lim_{N\to\infty}G_N$. 
We recall that $\deg(\Gamma^{(k)}_{pp}) = p-3$ and~$\Gamma_{pp}^k$ has $2p$ edges. Since~$p \geq 4$ in our case, we have
\begin{equation}
\abs{\PiNBPHZ \Gamma^{(k)}_{pp}} \lesssim K^{2p} \quad \text{for all} \quad p \geq 4
\end{equation}
and thus obtain the bound
\begin{equation}
\abs{(\PiNBPHZ\circ\cP) X^p}
\lesssim 
K^{2p} \sum_{k} b_{pp}^{(k)}
\lesssim
K^{2p} (4p-1)!! \;.
\end{equation}
Note that we have bounded $\sum_{k} b_{pp}^{(k)}$, the number of possibilities to pair the $4p$ vertices of~$X^p$ 
\begin{enumerate}[label=(\arabic*)]
\item without self-interactions and 
\item such that the resulting Feynman diagram is connected
\end{enumerate}
by disregarding these two constraints, which gives~$(4p-1)!!$ possibilities.
Finally, we have
\begin{equation}
\abs{(\PiNBPHZ\circ\cP)S_n}
\lesssim
\sum_{p=4}^{n-1} \alpha^p K^{2p} \frac{(4p-1)!!}{p!} 
\asymp
\sum_{p=4}^{n-1} \Order{p!}  \eps^p\;,
\end{equation}
where the last asymptotic equality follows similarly as in the $\Phi^4_0$ case, cf.~\eqref{eq:Zeps}.
\end{proof}

Finally, the formula for the Taylor remainder of~$\e^{-\alpha X}$ in conjunction with the intermediate value theorem implies the existence of some~$\theta \in [0,1]$ such that
\begin{equation}
R_n = \frac{(-\alpha)^n}{n!} X^n \e^{-\alpha \theta(X) X}\;.
\label{eq_remainder}
\end{equation}
Therefore, we are now ready to give the proof of the main result in this subsection.

\begin{proof}(of Proposition~\ref{prop:borel_dimthree})
We aim to apply Sokal's theorem. Observe that
\begin{enumerate}[label=(\arabic*)]
\item Lemmas~\ref{lem:subtraction_divergent} and~\ref{lem:asymp_Sn} prove the first condition~\eqref{eq:asympt_Sokal} and
\item the equality in~\eqref{eq_remainder} together with the moment bound assumption~\eqref{prop:borel_dimthree_moment_bound} establish~\eqref{eq:bound_Sokal}.
\end{enumerate}
Since we have assumed analyticity of~$\frac{Z_N(\eps)}{Z_N(0)}$, we conclude by Sokal's result, Theo\-rem~\ref{thm:sokal}.
\end{proof}

\bibliographystyle{plain}
{\small \bibliography{BK}}


\vfill

\bigskip\bigskip\noindent
\begin{minipage}[t]{.5\textwidth}
\small
\textbf{Nils Berglund} \\
Institut Denis Poisson (IDP) \\ 
Universit\'e d'Orl\'eans, Universit\'e de Tours, \\ CNRS -- UMR 7013 \\
B\^atiment de Math\'ematiques, B.P. 6759\\
45067~Orl\'eans Cedex 2, France \\
{\it E-mail address: }
{\tt nils.berglund@univ-orleans.fr}
\end{minipage}%
\hspace{2em}
\begin{minipage}[t]{.45\textwidth}
\small
\textbf{Tom Klose} \\
University of Warwick \\ 
Department of Statistics \\
Coventry CV4 7AL, United Kingdom \\
{\it E-mail address:}
{\tt tom.klose@warwick.ac.uk}
\end{minipage}

\end{document}